\documentclass[12pt,halfline,a4paper]{ouparticle}
\usepackage{parskip}
\usepackage{algorithm}
\usepackage{algpseudocode}
\usepackage[utf8]{inputenc}
\usepackage[colorlinks=true, citecolor=black]{hyperref}
\usepackage{threeparttable}
\usepackage{booktabs}
\usepackage{blindtext}
\usepackage{authblk}
\usepackage{graphicx}
\usepackage{rotating}
\usepackage{amsthm}
\usepackage{multirow}
\usepackage{hyperref}
\usepackage{makecell} 
\usepackage{lipsum}

\usepackage{listings} 

\usepackage[labelfont=bf]{caption}
\usepackage[labelformat=simple]{subcaption}

\usepackage[authoryear,round]{natbib}
\hypersetup{
    colorlinks=true,
    linkcolor=blue,
    filecolor=magenta,      
    urlcolor=cyan,
    pdftitle={Overleaf Example},
    pdfpagemode=FullScreen,
    }

\newtheorem{proposition}{Proposition}

\begin{document}

\title{Open Sourcing GPTs: Economics of Open Sourcing Advanced AI Models\thanks{I wish to express my sincere gratitude to my advisors: Francesco Decarolis, Avi Goldfarb, Andrea Fosfuri, and Carlo Schwarz. I am also grateful to Kevin Bryan, Alfonso Gambardella, Joshua Gans, and Claudio Panico whose suggestions greatly helped to improve the paper. Please address correspondence to  mahyar.habibi@phd.unibocconi.it.}}

\author{
\name{Mahyar Habibi}
\address{Department of Economics, Bocconi University}
}

\abstract{This paper explores the economic underpinnings of open sourcing advanced large language models (LLMs) by for-profit companies. Empirical analysis reveals that: (1) LLMs are compatible with R\&D portfolios of numerous technologically differentiated firms; (2) open-sourcing likelihood decreases with an LLM's performance edge over rivals, but increases for models from large tech companies; and (3) open-sourcing an advanced LLM led to an increase in research-related activities. Motivated by these findings, a theoretical framework is developed to examine factors influencing a profit-maximizing firm's open-sourcing decision. The analysis frames this decision as a trade-off between accelerating technology growth and securing immediate financial returns. A key prediction from the theoretical analysis is an inverted-U-shaped relationship between the owner's size, measured by its share of LLM-compatible applications, and its propensity to open source the LLM. This finding suggests that moderate market concentration may be beneficial to the open source ecosystems of multi-purpose software technologies.
\par 
\vspace{5mm} 

} 

\author{Mahyar Habibi \break{\normalsize{ Department of Economics, Bocconi University}}}

\date{\today}

\maketitle
\thispagestyle{empty}

\newpage
\setcounter{page}{1}
\section{Introduction}
Open source contributions have significantly shaped the growth of artificial intelligence, machine learning, and more recently large language models (LLMs). Interestingly, large for-profit technology companies have played crucial and at times dual roles in this rapidly evolving landscape. On one hand, these companies have made notable contributions to the open source ecosystem by sharing scientific breakthroughs such as Transformer architecture and open-sourcing advanced software like TensorFlow, PyTorch, and LLaMA. The extent and impact of their contributions over the past decade arguably surpass those made by the most prolific academic institutions \citep{ahmed2023growing}. On the other hand, following recent breakthroughs in LLM capabilities, some major technology firms have revised their stance toward the open source ecosystem. They now restrict and monetize access to their LLMs while expressing concerns about the dangers of open-sourcing advanced models \citep[e.g.,][]{de_vynck_2023, Lin2024ShouldAI}.

This paper argues that open-sourcing advanced AI models like LLMs presents profit-maximizing firms with a strategic trade-off between accelerating technological growth and securing immediate financial returns. The key takeaway is that firms are most likely to open source multi-purpose software such as LLMs when they own a significant but not excessive share of compatible applications. Small firms with few compatible applications prefer a closed strategy for immediate revenue, while firms dominating compatible applications find open source community contributions insignificant compared to their internal resources. However, for technologies with wide-ranging use cases like LLMs, even Big Tech giants own a modest share of compatible applications, potentially finding the benefits of open sourcing outweigh the costs. Meta CEO Mark Zuckerberg's remarks on Generative AI and the company's LLaMA open sourcing strategy align with this argument. Zuckerberg stated, ``In the last year, we have seen some really incredible breakthroughs — qualitative breakthroughs — on generative AI and that gives us the opportunity to now go take that technology, push it forward, and \emph{build it into every single one of our products}," and while he does not expect LLaMA to generate ``a large amount of revenue in the near term, but over the long term, hopefully that can be something" \citep[][]{vanian2023zuck, vanian2023metas}.  This insight contributes to our understanding of the economics of open sourcing in AI, highlighting how the properties of AI as a potential general-purpose technology influence firms' strategic decisions and, consequently, the AI development trajectory. 

The analysis proceeds in four parts. The first part examines the compatibility of LLMs in the R\&D process of innovating firms. For this part, I use patent data and propose a novel strategy to examine compatibility of firms' R\&D process with LLMs. I find that LLMs are compatible with R\&D portfolios of a large set of technologically diverse firms, implying a broad range of industrial applications for this technology. 

In the second part, I examine the relationship between the quality of the models and the open-sourcing strategy of the developers. Using data on major model releases and their performance on a widely used benchmark, I find that a 10-point increase in quality (on a 100-point scale) over the existing state-of-the-art open source model is associated with a 10-11 percentage point decrease in the likelihood of the model being open sourced. Furthermore, for-profit organizations are, on average, 14-18\% less likely to open source a model. However, the analysis suggests that Big Tech companies, ceteris paribus, are 20\% more likely to open source a model than other for-profit organizations. In the third part of the analysis, I combine AI/ML-related publication records with GitHub data and document a significant increase in research-related activities among LLM researchers following the open source release of LLaMA, an advanced LLM developed by Meta. This finding implies that open-sourcing advanced software can stimulate related R\&D efforts.

Motivated by these findings, I propose a theoretical framework in the final part of the analysis to examine the decision-making process of for-profit firms in developing and open-sourcing a new LLM. In the theoretical analysis, LLMs are framed as a potential general-purpose technology (GPT), capable of boosting profits in various applications. The model is structured as a two-stage decision-making process. Initially, a firm assesses the quality of the existing open source model to decide whether to develop a new LLM. In the second stage, the firm decides how to optimally allocate computational resources for integrating the model into its applications. Should the firm opt to develop a new LLM, it then faces a choice: permanently open source the model or keep it proprietary for an additional period. This decision presents a strategic trade-off: stimulate software growth and R\&D efforts through open-sourcing or secure immediate profits by licensing. By open-sourcing, a firm leverages external contributions to enhance the model, accelerating its growth and integrating it more effectively with applications to boost profits. Alternatively, a closed-source release enables immediate revenue through API sales to external software producers, at the expense of missed community contributions.

The theoretical analysis generates several key predictions aligned with empirical findings. It suggests that the open-sourcing decision depends strongly on the quality lead over alternative open source models, with larger leads favoring closed-source strategies. The analysis predicts an inverted-U shaped relationship between firm size and open-sourcing tendency, reflecting varying benefits from accelerated growth at different scales of LLM-compatible applications. Additionally, while both small and large firms may find developing new LLMs profitable when existing open source quality is modest, only larger firms are likely to do so when high-quality open source alternatives exist. Moreover, the model reveals nuanced effects of open source ecosystem efficiency. In a strong ecosystem, open-sourcing a marginally superior model may be beneficial. However, as the quality gap widens, open-sourcing becomes less attractive and the firm may have incentives to limit the efficiency of the open source ecosystem, thereby slowing the progress of open source rivals. This insight is particularly relevant given recent calls from major tech companies to regulate open source releases of advanced models  \citep[e.g.,][]{de_vynck_2023, nolan2023big}.

AI is not the only field that saw significant contributions from for-profit companies to its open source ecosystem. Much of the infrastructure of Internet rests on foundations that were open sourced by for-profit firms, as well as operating systems for personal computers (Linux) and mobile devices (Android). Consequently, there's an extensive literature on the economics of open source software. This literature typically falls into two, sometimes overlapping categories. The predominant category examines programmers' motivations for contributing to open source projects. Though these incentives are crucial to the open source ecosystem of AI, my study does not delve into the individuals' incentives for contributing to open sourced AI projects. Instead, I focus on modeling the open-sourcing decisions of firms where a functioning open source community exists. For those interested in the incentives of open source contributors, \cite{lerner2002some} offers a comprehensive introduction to this area.

The second stream of literature on open source software, examines why firms choose to open source their proprietary software. This phenomenon extends beyond AI, with a history of strategic open-sourcing decisions in various for-profit sectors. Existing research predominantly identifies the attraction of users to complementary proprietary products as a key driver for open sourcing \citep[e.g.,][]{lerner2002some, hippel2003open, lerner2006dynamics, fosfuri2008penguin}. However, other motivations are also discussed. \cite{henkel2004open} discusses standard-setting and signaling technical prowess, while \cite{economides2006two} considers open-sourcing as a platform strategy to benefit from proprietary applications built upon it. \cite{gambardella2018open} points out that downstream firms may collaborate on open source alternatives to bypass upstream suppliers, and \cite{nagle2018learning} highlights the learning benefits firms gain from crowd feedback in open source projects.  The theoretical framework in this study draws parallels to the competition between for-profit and non-profit entities in operating systems in \cite{casadesus2006dynamic}, where the focus is on demand-side learning.

I make two contributions to this strand of literature. Firstly, despite being frequently discussed in the literature \citep[e.g.,][]{lerner2002some, lerner2006dynamics}, empirical evidence concerning the impact of open source software on encouraging research activities in a causal framework is rare. To the best of my knowledge, this study is the first to document empirical evidence concerning the potential effects of open source software on stimulating research activities. \cite{nagle2019open} studies the impact of using open source software on firms' productivity and finds a positive and significant impact on the subset of firms with an ecosystem of complements. However, I am not aware of a study that directly investigates the impact of open source on research activity within a causal framework. Secondly, this study departs from existing literature by treating software not just as a product but as an enabling technology with applications across various sectors, generating nuanced insights into strategic development and open sourcing decisions not fully captured by existing frameworks. 

Instances of inventors sharing technological advancements openly are rare, but not exclusive to AI. This phenomenon, termed ``collective invention'' by \cite{allen1983collective}, was observed in 19th-century iron-making in Britain’s Cleveland district, where companies freely exchanged blast furnace design improvements. Similar patterns emerged in post-1800 steam engine enhancements \citep{nuvolari2004collective} and the flat panel display industry’s evolution \citep{spencer2003firms}. \cite{osterloh2007open} further suggest that open source software development is a modern embodiment of this collective invention concept. The open source ecosystem in AI and LLMs shares similarities and differences with historical collective invention cases. A common thread is the reliance on experimental trial and error, where shared experiences significantly enhance learning opportunities. However, in contrast to the AI ecosystem, where large tech companies play a pivotal role, historical episodes of collective invention often featured smaller firms with limited R\&D resources. This study proposes that open source contribution of tech firms in AI is attributable to the broad applicability of AI, extending beyond the scope of any single firm. Consequently, the opportunity for each major tech firm to leverage community resources for the rapid advancement of their models remains substantial. 

This study also relates to recent work examining the changing dynamics between industry and academic research. \cite{arora2020changing} and \cite{arora2021knowledge} document how corporate labs have shifted away from basic research towards development activities, potentially hindering the emergence of general-purpose technologies. They argue that firms' scientific research decisions are shaped by a trade-off between internal benefits and spillover costs to rivals, suggesting this dynamic has contributed to declining corporate research investment. My analysis suggests that the tension between knowledge spillovers to rivals and appropriability may be partially mitigated when the technology's application domain is sufficiently expansive and firms can protect their competitive advantage through downstream specialization, offering a new angle to understand open sourcing advanced AI systems by big tech companies.

This paper also contributes to the rapidly growing field of the economics of AI. A growing strand of literature focuses on AI and more recently LLMs characteristics as a general-purpose technology \citep[e.g.,][]{brynjolfsson2018artificial, cockburn2018impact,agrawal2023artificial, agrawal2023similarities, goldfarb2023could, eloundou2023gpts}. Beyond the analysis of AI as a GPT, \cite{jacobides2021evolutionary} and \cite{ahmed2023growing}  highlight the dominance of few Big Tech firms in terms of resources and influence on AI research. The role of open source in AI is further examined by \cite{rock2019engineering}, studying how open-sourcing TensorFlow by Google affected the market valuation of AI-focused companies. This study contributes to this literature by exploring how  characteristics of LLMs, as a potential general-purpose technology, influence firms' decisions to open source their models, and consequently, the technology's development trajectory. 

Lastly, the method proposed in this paper for obtaining latent technology representation of firms and technologies can contribute to the broader innovation literature interested in examining similarities and differences in R\&D processes using patent data. Recently, there has been growing interest in using unsupervised NLP techniques to represent a firm’s R\&D portfolio within a latent vector space \citep[e.g.,][]{arts2021technology, hain2022text}. However, the popularity of unsupervised techniques in AI/ML is primarily driven by the unavailability of enough labeled training data by domain experts \citep{hovy2022text}. This contrasts sharply with patent data, where patents are classified by domain experts into comprehensive and detailed patent classification systems. Although efforts to use patent classification systems to represent firms' technologies go as far back as \cite{jaffe1986}, the challenges posed by the discrete and rigid structure of patent classification systems have encouraged researchers to adopt unsupervised techniques for these purposes. Inspired by classical NLP and ML techniques, I propose a flexible method that overcomes these challenges and creates a technology latent space using ``gold standard'' data without relying on fully unsupervised techniques.

The remainder of this paper is structured as follows: Section \ref{sec:context} provides a brief overview of the ecosystem of LLMs. Section \ref{sec:data} describes the data. Section \ref{sec:landscape} introduces the method used to create the latent technology space and analyzes LLMs within the constructed technology landscape. Section \ref{sec:empirics} studies the open-sourcing decisions in the LLM ecosystem and examines impact of open-sourcing LLaMA on research activity of LLM-researchers. Section \ref{sec:model} introduces the theoretical framework and outlines its predictions, and Section \ref{sec:conclude} concludes the paper.

\section{Context; LLMs in a Nutshell}
\label{sec:context}
Although not clearly defined, Large Language Models (LLMs) can be broadly described as models based on artificial neural networks with billions of parameters, trained on a vast amount of text data in an unsupervised fashion, and capable of processing and often generating natural language data. In practice, however, LLMs mostly refer to models, often with tens of billions of parameters, built on the \emph{Transformer} architecture proposed by \cite{vaswani2017attention}. The first generation of models now commonly recognized as LLMs, including BART, GPT-2, and T5, were released in 2019 \citep[]{lewis2019bart, radford2019language, raffel2020exploring}. Since then, there has been a significant increase in the quality, quantity, and scale of such models. This section briefly describes the development process and the open source ecosystem of LLMs.

Large Language Models often undergo two main phases during their development: pre-training and fine-tuning. During pre-training, the model is exposed to a vast amount of text data (often tens of Terabytes) and is trained to predict the next word in a sentence given the previous words. This process requires massive computing resources (thousands of GPUs), can take months, and costs tens of millions of dollars for the state-of-the-art models \citep{NYT2023GPT}. The result is a versatile model capable of generating coherent text but often unable to provide desired responses for specific applications (e.g., chat bot). The purpose of fine-tuning is to adapt the pre-trained model to a specific task or domain and involves updating the parameters of the pre-trained model on a smaller, task-specific dataset \citep{howard2018universal}. It is noteworthy that fine-tuning, often done using relatively small datasets (e.g., 10-100K examples), incurs costs that are a fraction of the pre-training costs. The fine-tuning step can also involve additional steps such as reinforcement learning from human feedback (RLHF), which integrates human judgments directly into the fine-tuning process and allows models to learn preferences that are difficult to capture with traditional datasets or reward structures.

In the domain of LLMs, open source is more nuanced than what is traditionally perceived as open source software (OSS). OSS can be loosely defined as software projects with published source code accompanied by a license allowing modification and redistribution \citep{lerner2002some}. However, for LLMs, the training code is neither the only nor the most critical component of the software. The key components enabling practical use are the model parameters or weights. These weights can be made public without specifying the full training procedure or the model's architecture. Additionally, datasets for fine-tuning the model can either be open sourced or kept proprietary. Furthermore, there are stark differences among contributors within the open source ecosystem of LLMs. A model's performance, in terms of next-word prediction accuracy, depends largely on the model's size, dataset volume, and computing resources dedicated to training \citep{kaplan2020scaling}. This so-called ``scaling laws'' of language models implies that developing state-of-the-art LLMs from scratch incurs significant costs, limiting the ability of many organizations to contribute new pre-trained models to the open source community. However, when it comes to open-sourcing datasets, new training or inference methods, or releasing fine-tuned models, the open source community is more diversified.

Figure \ref{fig:orgs} illustrates the number of open and closed models for ten leading organizations according to the ecosystem dataset from the Center for Research on Foundation Models (CRFM) at Stanford University\footnote{This dataset includes only text-based models as well as multi-modal models such as text to image or text to audio.}.  Google, OpenAI, Microsoft, and Meta lead in the number of models released. While most organizations have released both open and closed models, their strategies for open-sourcing vary significantly. Prominent AI startups such as OpenAI, Cohere, and Anthropic tend to keep their models primarily closed. Among the Big Tech companies, Meta has released more models publicly, and its LLaMA-series models are among the largest and most widely used open models. Google, on the other hand, employs a different strategy by keeping its flagship and larger models closed, while continuing to open source smaller models.

\begin{figure}[htp]
\centering
\footnotesize
\caption{The Number of Open and Closed Models by Leading LLM Developers \label{fig:orgs}}
    \centering
        \includegraphics[width=.9\linewidth]{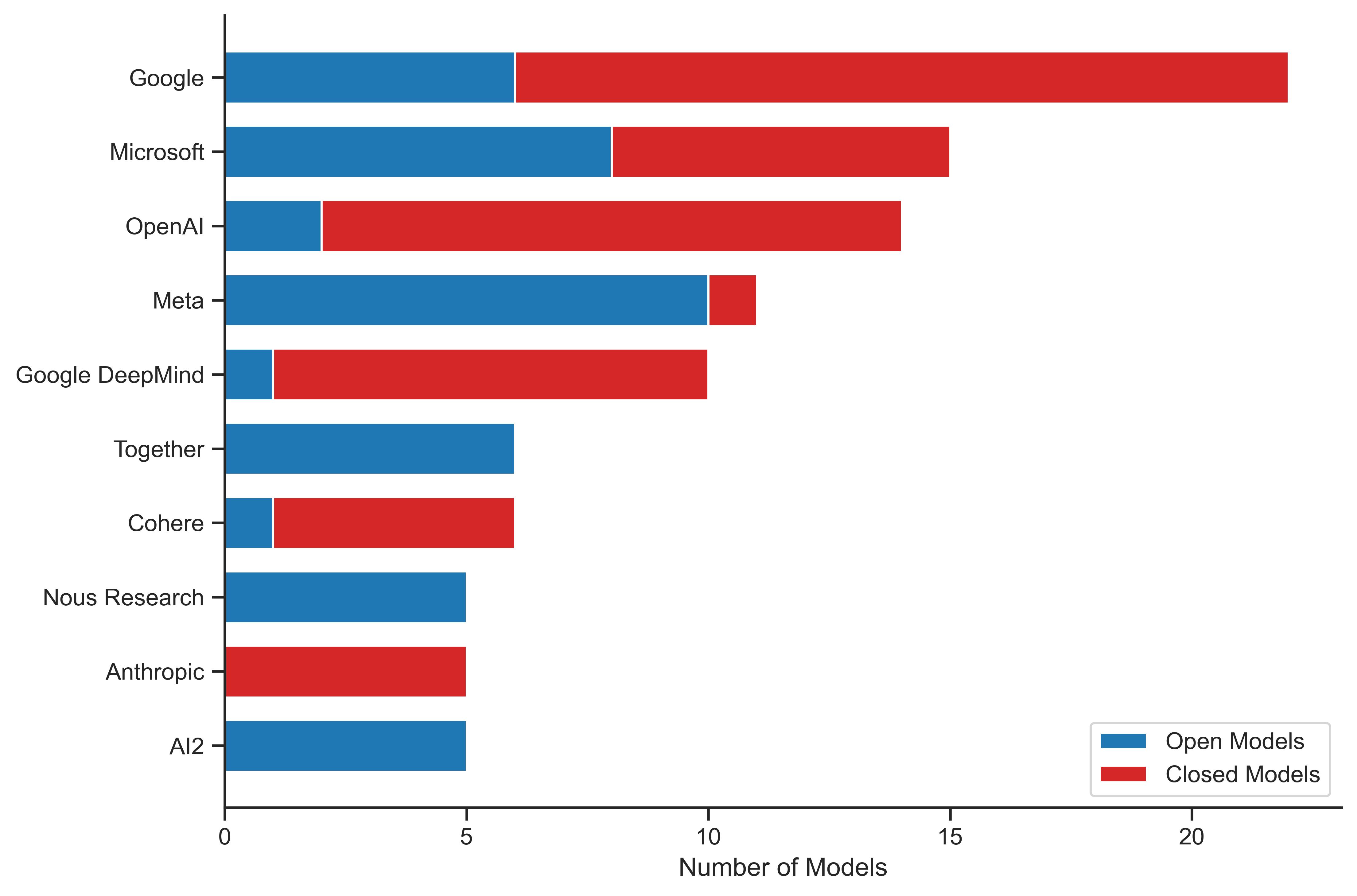}
        
    \hspace{0.4cm}\parbox{\textwidth}{\footnotesize{\textit{Notes:} The figure illustrates the number of open and closed LLMs by ten leading LLM developers in the ecosystem dataset of Center for Research on Foundation Models at Stanford University.}}
\end{figure}

\section{Data}
\label{sec:data}

This study collects data from multiple sources for a comprehensive analysis at the firm and researcher levels within the open source (OS) ecosystem of large language models (LLMs). Data were extracted from \emph{Papers-with-Code}\footnote{Paperswithcode.com}, \emph{arXiv}, and \emph{GitHub} to construct proxies for firms' contributions to the open source community and activities of LLM researchers. Additionally, an analysis of patent application data filed with the US Patent and Trademark Office (USPTO) was conducted to evaluate the technology compatibility of firms and their engagement with Foundation Models. Below, a detailed description of each data source is provided.

\begin{itemize}
    \item \textbf{Papers-with-Code} is a community-driven initiative led by the core team at Meta AI Research. It provides practitioners with free access to AI/ML research resources. This platform maintains up-to-date information on open-access AI/ML publications and tracks the presence of both official and unofficial code repositories associated with each paper. I retrieved the data in January 2024, focusing on papers that have an official repository on GitHub and were published on arXiv in 2019 or later. The resulting dataset contains more than 108 thousand publications.
    
    \item \textbf{arXiv}: Data on the initial publication dates, titles, and abstracts of papers were collected from arXiv. This analysis focused on papers published from 2019 onward, coinciding with the release of the first generation of LLMs  such as GPT-2 and T5. Identifying papers related to LLMs was based on the analysis of its title and abstract, the methodology of which is detailed subsequently in this section.

    \item \textbf{GitHub}:  I consider an open-access paper listed on Papers-with-Code with an official GitHub code repository as an open source contribution.\footnote{The presence of an official code repository differentiates papers that contribute tangibly to the open source ecosystem from those that only describe model performance across various benchmarks and tasks without any open source contribution.} I extract two critical pieces of information from GitHub: information on repository owners and  data of contributors to these repositories. 
    
    \item \textbf{USPTO}: This study employs patent data to identify firms that utilize generative language models in their R\&D efforts, to assess the compatibility of a firm's technology with LLMs, and to evaluate the breadth of applications firms aim to integrate with this technology. Considering LLMs' status as an emerging technology, patent application data was preferred over granted patent data because the latter captures innovation activities only after a delay of at least a few years. The data were accessed through PatentsView.org in February 2024, providing information updated through December 31, 2023. The analysis is confined to utility patent applications filed by organizations with at least two applications during 2019-2023. The dataset comprises nearly 1.4 million applications from c.a. 61 thousand organizations. 
\end{itemize}
\subsubsection*{Merging the Datasets}
The Papers-with-Code dataset includes URL links to the respective arXiv pages and GitHub repositories, provided there is an associated page or repository for the paper. Linking organizations that own these repositories with those listed in patent application data is less straightforward. A primary challenge arises because the GitHub data often represent research groups within various institutions. Typically, details about these organizations are available on the organizations' biography pages; however, this information is largely unstructured. To tackle this, I employ a language model to parse the information and identify profiles associated with commercial entities. After cleansing the names of applicants and repository owners, I merge the datasets using exact matching on cleaned names and unique-part matching for the remaining subset, subsequently removing false matches through manual inspection. I further analyze unmatched organizations with high string similarity for potentially overlooked matches, adjusting the organization names in both datasets for initial match compatibility. Ultimately, approximately 180 organizations across the two datasets were successfully linked. For further details regarding the matching procedure and data parsing with the language model, see Appendix \ref{sec:app_data}.

\subsubsection*{Identifying LLM-Related Papers}
 To identify papers related to LLMs, a narrow keyword search was deemed insufficient due to the rapidly evolving technical vocabulary in the field, which could either omit relevant papers or yield excessive false matches. To address this challenge, I fine-tuned an LLM classifier specifically for identifying LLM-related papers through a two-step process.  Initially, two commercial LLMs, GPT-3.5 and Mixtral 8x7B, annotated a set of 20,000 out-of-sample papers \footnote{This sample originated from papers associated with unofficial code repositories on GitHub, in contrast to the main analysis focusing on papers with official repositories. To ensure reproducibility, the models' temperatures were set to zero.}. Two separate models were employed to mitigate the reliance on a single model's classification outcome. The annotated dataset then served to fine-tune the pre-trained SciBERT language model  \citep{beltagy2019scibert} for this specialized task, achieving accuracy of 0.97 and an F1 score of 0.78 on a holdout sample. Finally, I used the fine-tuned model to identify LLM-related papers in the main sample.

\subsubsection*{Patent Applications} 
I leverage the Cooperative Patent Classification (CPC) system to identify firms integrating generative language models into their innovation activities. Additionally, I use applicant information to link organizations in the patent application data with those in the GitHub dataset, containing firms contributing to the open source ecosystem of LLMs. Furthermore, I apply the CPC system to gauge firms' compatibility with LLMs and the scope of their R\&D activities involving this technology, a process I will detail in Section \ref{subsec:lta}.

\section{Measuring the Scope of Application of LLMs}
\label{sec:landscape}

The primary objective in this section is to study the scope of industrial applications of LLMs through the lens of patent data. By analyzing technological differences among firms that could potentially leverage LLMs in their R\&D processes, I aim to develop a better understanding of the environment in which open-sourcing decisions take place. To this end, I first outline the method I developed to assess the compatibility of firms' technologies with LLMs in a latent technology vector space. The analysis suggests a large number of firms in the patent application dataset have high compatibility with LLMs. Moreover, I find  significant variation in the R\&D portfolios of firms with potentially LLM-compatible technologies, suggesting a broad spectrum of industry applications for LLMs. I also examine data on the leading for-profit contributors to the open source ecosystem for LLMs and find that the majority of LLM applications extend beyond the R\&D scope of any single firm. These insights into the broad applicability of LLMs  will form a cornerstone of the theoretical framework introduced in Section \ref{sec:model}.

\subsection{Crafting the Latent Technology Space}
\label{subsec:lta}
In this subsection, I present a novel methodology to create a latent technology space by leveraging the richness of patent classification systems and the flexibility of ML techniques. The goal is to represent each firm's overall R\&D portfolio and LLMs in a high-dimensional vector space. The distance between a firm's vector and LLMs' vector will be used as a proxy for the firm's R\&D compatibility with LLMs technology. 

Mapping firms' technological positions in a vector space through patent classification has been a longstanding practice \citep[e.g.,][]{jaffe1986}. Nevertheless, the discrete and hierarchical nature of the patent classification constrains the capability of traditional methodologies to capture nuanced technological profiles of firms. For instance, Jaffe's seminal method employed an  ``ad hoc'' categorization of over 300 patent classes into 49 groups, a schema likely too coarse to discern subtle technological distinctions. To address these shortcomings, researchers have employed NLP techniques to construct more refined vector spaces from patent texts \citep[e.g.,][]{arts2021technology, hain2022text}.

However, using NLP techniques to represent technologies presents at least two major drawbacks. First, unsupervised approaches fall short of expert labeled data in capturing high quality information in complex tasks \citep{hovy2022text}. This issue is  exacerbated with new technologies, where there might not be sufficient training data available about the new technology to enable the models to create a accurate representation of the technology in an unsupervised manner\footnote{For example, in the case of LLMs, "Transformer" refers to a revolutionary architecture proposed by \cite{vaswani2017attention}. However, "transformer" can also describe the widely used electrical device for changing voltage when transferring electric energy from one alternating-current circuit to another. A subject matter expert can immediately differentiate the two upon first encounter, whereas a language model requires a substantial amount of data to capture the differences.}. The second drawback concerns resources and efficiency. Even basic NLP techniques, create challenges for researchers when applied to a large volume of patent data \citep[e.g.,][]{kelly2021measuring}. 
 
The proposed method is detailed in Algorithm \ref{alg:lta}. Initially, the method constructs a rich representation of patents by leveraging the hierarchical structure in the Cooperative Patent Classification (CPC) system\footnote{The same methodology can be applied to other popular patent classification systems, including IPC and WIPO.}. For example, consider CPC code G06F40, which denotes the handling of natural language data. This code breaks down into section G, denoting Physics; subsection G06, specifying Computing, Calculating, or Counting; and class G06F, representing Electric Digital Data Processing. Typically, a patent is classified with multiple such codes. The method converts this hierarchical structure into a flattened representation by capturing higher-order interactions among codes at the same level, enhancing the representation's richness. For instance, a patent classified with G06F40 and H04W4 is represented as [G, H, G-H, G06, ..., G06F40, H04W4, G06F40-H04W4]. This technique is analogous to incorporating \emph{n-grams} in the \emph{Bag-of-Words} representation of textual documents \citep{gentzkow2019text}. The subsequent step aggregates the patents at the firm level, creating a firm-token matrix where a firm's overall R\&D is represented by the frequency of each token (e.g., G-H) in its patent (application) portfolio.

The constructed firm-token matrix, being sparse and high-dimensional, is not immediately conducive to depicting firms' technologies. Subsequently, dimensionality reduction, following row-normalization, compresses this sparse representation into a denser, lower-dimensional matrix \footnote{In this exercise, I utilized the first four levels of the CPC system to represent firm technologies, for instance, up to G06F40 as demonstrated in the previous example. I included only applicants with at least five applications between 2019 and 2023, resulting in nearly 1.25 million patent applications. Furthermore, I discarded code combinations occurring less than five times due to their rarity. After aggregation at the firm level, the resulting matrix contained nearly 24,000 rows (each representing an applicant) and over 260,000 columns, each corresponding to a technology code combination. Ultimately, I applied truncated SVD to reduce the original matrix to 512 dimensions. This reduced matrix accounts for 69.7\% of the variance in the original matrix.}. This process builds on the classic Latent Semantic Analysis technique, introduced in \cite{deerwester1990indexing}.

\begin{algorithm}[htp]
\renewcommand{\labelenumii}{\arabic{enumi}.\arabic{enumii}}
\renewcommand{\labelenumiii}{\arabic{enumi}.\arabic{enumii}.\arabic{enumiii}}

\caption{Creating Latent Technology Space}\label{alg:lta}
\begin{enumerate}

    \item Specify $D$, the desired level of depth in the hierarchical patent classification system.
    \item Specify $n$, the max order of interaction among classification codes in the same level of hierarchy.
    \item \textit{for} $d$ in $\{1,...,D\}$ :
    \begin{enumerate}
        \item Define the set of tokens by interacting classification codes up to the $n$-th order;  
        \item Create the patent-token counts matrix;
        \item Aggregate the counts matrix at the applicant level; store for the next step; 
    \end{enumerate}
    \item Concatenate and normalize the applicant-count matrices.
    \item Apply dimensionality reduction.
    \item (optional) For a particular technology:
    \begin{enumerate}
        \item Find the related patents.
        \item Create the patent-token matrix of counts. 
        \item Sum across all patents. 
        \item Apply the transformation used in step 5.
    \end{enumerate}
\end{enumerate}

\end{algorithm}

The next goal is to determine the position of LLM technology within the created latent technology space. For this purpose, all patents citing \cite{vaswani2017attention}\footnote{This citation data is available only for granted patents.}, which introduced the Transformer architecture, a fundamental building block of LLMs, were collected. I then identified a subset of these patents associated with CPC code G06F40, which denotes handling of natural language data, and treated these as LLM-related patents. These patents were aggregated as though filed by a single hypothetical firm and were projected onto the latent technology space using the previously acquired compression transformation.

Figure \ref{fig:lta_map} illustrates a 2-D representation  of the technology space\footnote{UMAP package was used to visualize the constructed vector space. To improve the illustration, only firms 100 applications or more are displayed in the figure.}. As a sanity check to verify the proposed method is effective in capturing technology similarities among firms, the figure marks ten well-known pairs of firms with similar R\&D portfolios, including Airbus and Boeing, AstraZeneca and Pfizer, and Ford and General Motors. As shown in the figure, these paired firms are positioned in close proximity to one another on the map. Additionally, the figure plots the location of LLM-related patents. As expected, major technology companies such as Amazon, Microsoft, and Google are all positioned close to the LLM technology on the map.
\begin{figure}[ht]
\centering
\caption{A 2-D Representation of the  Latent Technology Space
\label{fig:lta_map}}
    \centering
     \includegraphics[width=0.8\textwidth]{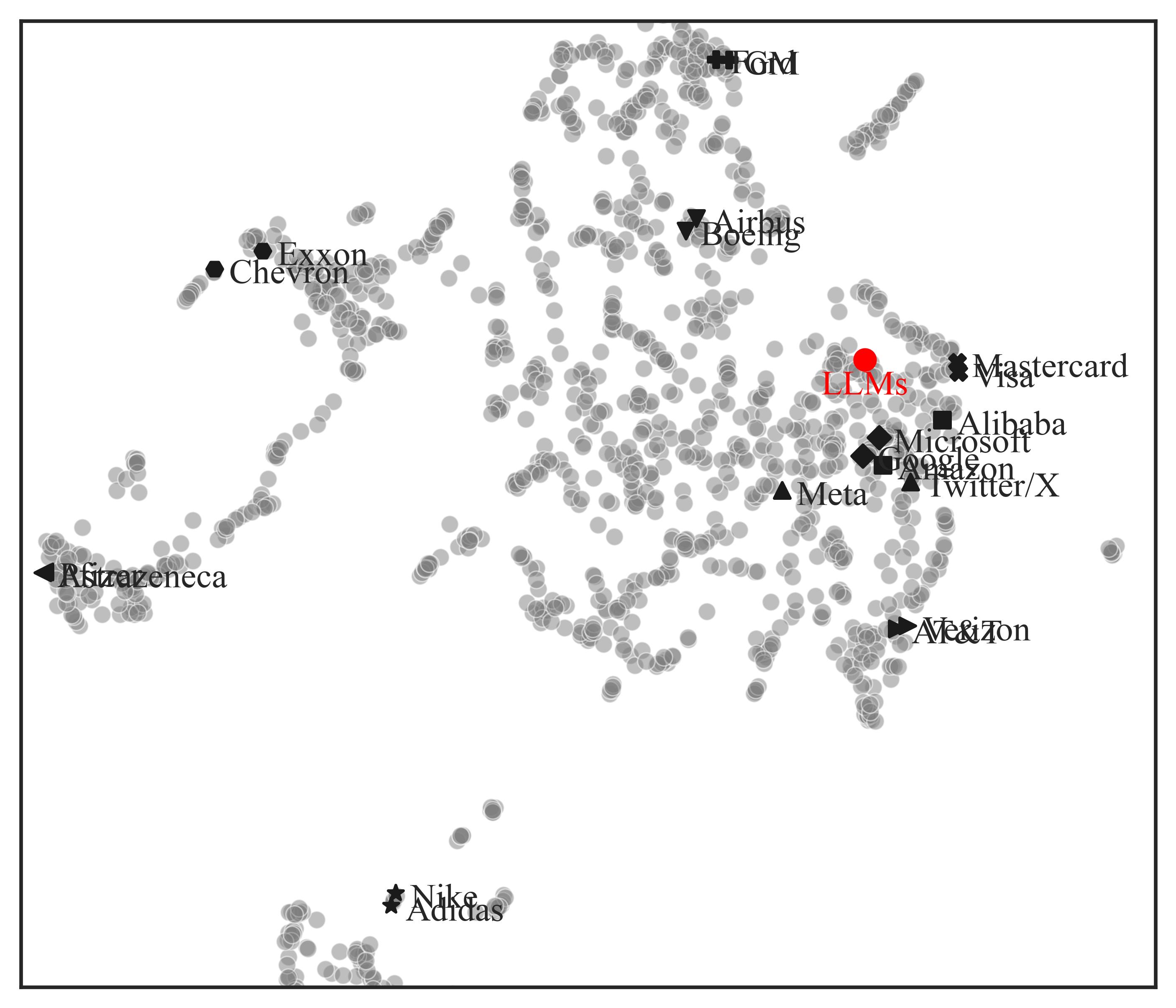}
    \hspace{0.4cm}\parbox{\textwidth}{\footnotesize{\textit{Notes:} The figure presents the projection of the constructed latent technology space in 2D. For improved illustration, only applicants with more than 100 applications are included, and a handful of outliers are omitted. Additionally, the figure plots 10 pairs of well-known firms with qualitatively similar technologies. A pooled portfolio of patents citing the ``Transformer'' paper \citep{vaswani2017attention} and related to natural language processing is represented by a red dot.}}
\end{figure}

\subsection{Firms' Compatibility with LLMs and open source Contributions}

Figure \ref{fig:tech_comp} displays the number of firms in the patent dataset that have an R\&D profile compatible with a selected subset of recent technologies, including LLMs. To obtain technology vectors (except for LLMs, whose technology vector was obtained earlier), patents in the dataset corresponding to the CPC code associated with each technology were collected\footnote{The CPC (Cooperative Patent Classification) codes corresponding to the technologies mentioned are as follows: Additive Manufacturing, B33Y10; Computer Vision, G06V10; Cosmonautics Vehicles, B64G1; Cryptocurrency, G06Q2220; Fusion Reactors, G21B; Mixed Reality, G06T19/006; Nanobiotechnology, B82Y5; Quantum Computing, G06N10; Robots, Y10S901.}. These patents were then processed and aggregated as if filed by a single entity and mapped onto the latent technology vector space, following a procedure similar to that described for LLMs. The compatibility of firms with each technology was assessed by cosine similarity between the firm's vector and the technology vectors, using a critical cosine similarity threshold of 0.7 to distinguish firms with an R\&D profile compatible with the technology from those that are not. This process identified LLMs, along with Computer Vision, Cryptocurrency, Mixed Reality, and Additive Manufacturing as technologies compatible with the R\&D processes of a relatively large number of firms. However, Quantum Computing, Fusion Reactors, Autonomous Robots, Spacecraft, and Nanobiotechnology were found to be compatible with a smaller subset of firms.

\begin{figure}[htb!]
\centering
\caption{Number of Firms with R\&D Profiles Compatible to Selected Technologies
\label{fig:tech_comp}}
    \centering
     \includegraphics[width=0.8\textwidth]{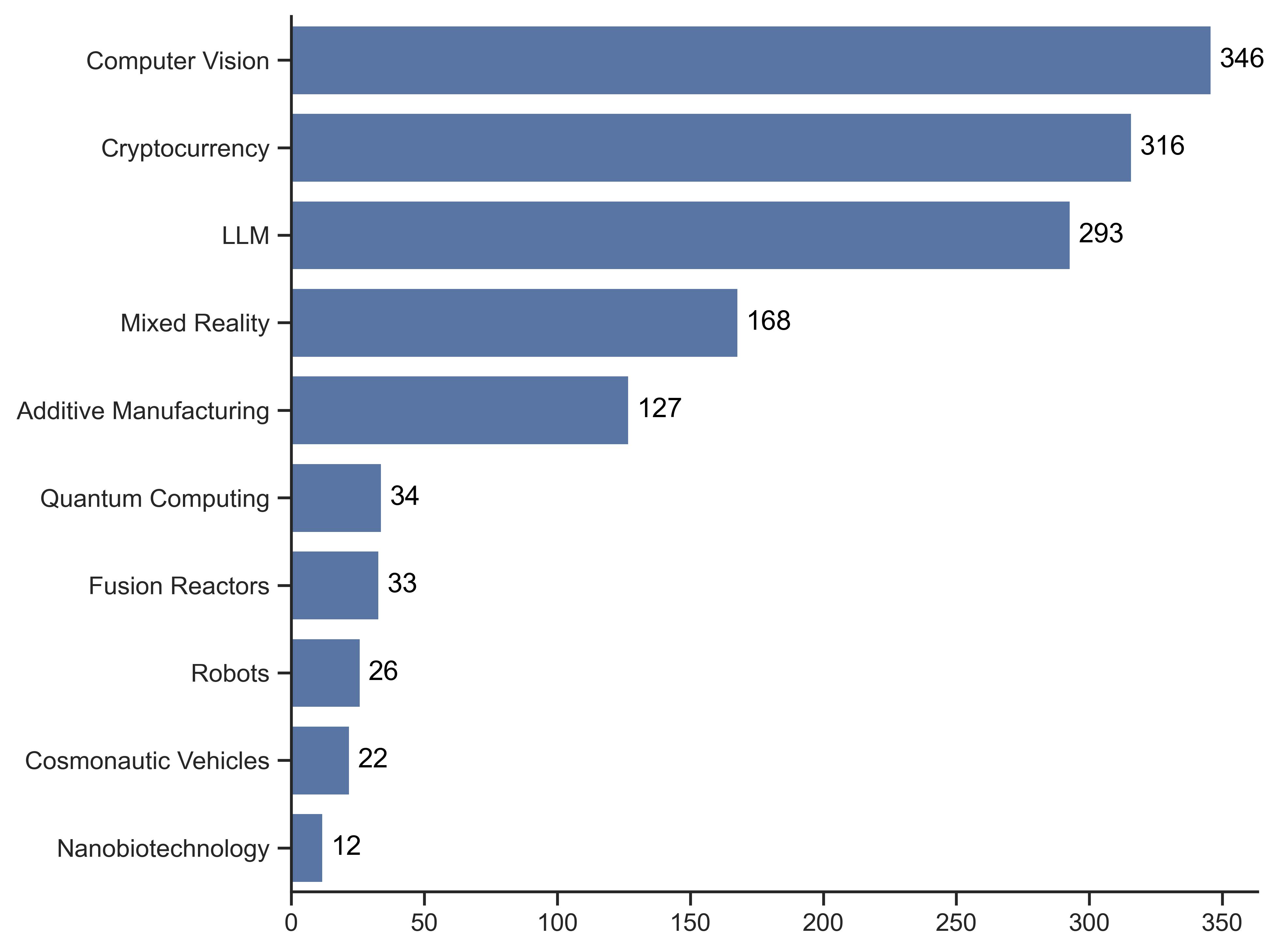}
    \hspace{0.4cm}\parbox{\textwidth}{\footnotesize{\textit{Notes:} The figure presents the number of firms in the patent application dataset with R\&D profiles compatible to the selected technologies.  A firm was considered to be compatible with a technology if the cosine similarity between its R\&D vector and the technology vector in the latent technology vector space  surpassed a critical threshold of 0.7. }}
\end{figure}

Figure \ref{fig:cross_sim} illustrates the distribution of cosine similarities between technology vectors of all pairs of firms with LLM-compatible R\&D portfolios, as identified in the previous exercise. The figure reveals substantial heterogeneity in the R\&D portfolios of firms with LLM-compatible technologies. Additionally, Figure \ref{fig:llm_sim} in the Appendix showcases 50 firms with the largest cosine similarities to LLMs in the latent technology space. Grammarly, an English writing assistance application, has the highest cosine similarity to this technology. The list also includes AI startups and established firms in various sectors, such as Accenture, Baidu, PwC, Thomson Reuters, and Xiaomi. Overall, these findings suggest that LLMs have a broad range of applications in industry, a key assumption for setting up the model in Section \ref{sec:model}.

\begin{figure}[ht!]
\centering
\caption{Distribution of Cross Cosine Similarities of Technology Vectors for LLM-Compatible Firms
\label{fig:cross_sim}}
    \centering
     \includegraphics[width=0.8\textwidth]{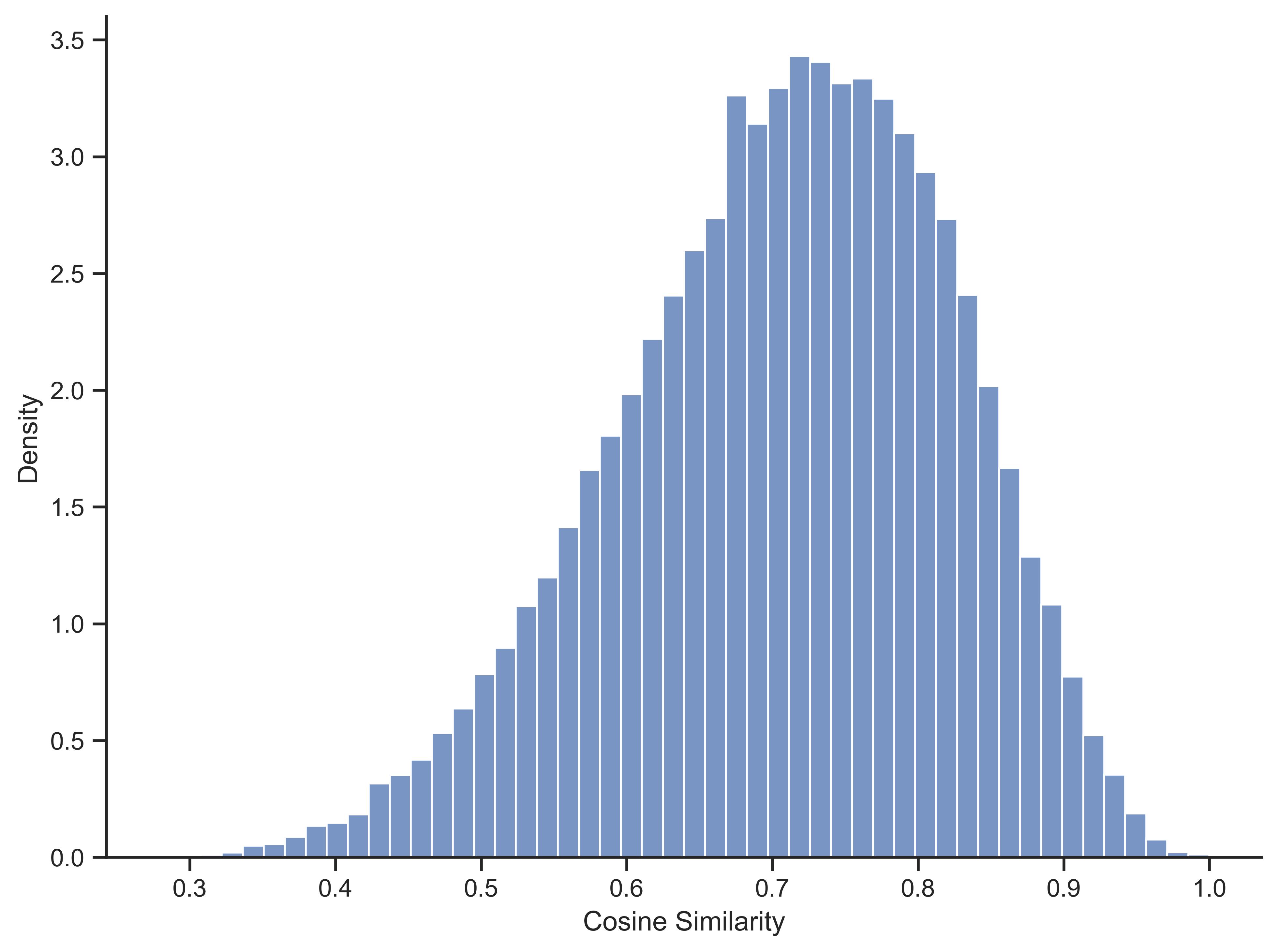}
    \hspace{0.4cm}\parbox{\textwidth}{\footnotesize{\textit{Notes:} The figure presents the distribution of cosine similarities between technology vectors of firms with LLM-compatible R\&D profiles. }}
\end{figure}

Table \ref{tab:app_repo} showcases ten companies with the most official repositories of LLM-related papers. The list includes five commonly recognized Big Tech firms: Microsoft, Google, Meta, Amazon, and Nvidia, as well as other prominent corporations including Alibaba, Salesforce, IBM, Intel, and Tencent. According to the proposed compatibility metric, all these firms have fairly high compatibility with LLMs, and yet none of them are ranked among the top 100 firms with the most LLM-compatible technologies. However, the vast R\&D portfolios of these firms, which include thousands of patent applications, imply a significant overall exposure to LLMs. Notably, IBM leads in the number of applications related to natural language generation across the dataset, with Google and Microsoft following in second and fourth places (trailing behind Capital One), and Meta taking the seventh rank. Another notable observation concerns the unique CPC codes within applications related to natural language generation. For instance, Microsoft's 34 applications in this domain encompass 141 unique CPC codes, constituting 12\% of all CPC codes in such applications. For IBM, which has the highest number of related applications, this proportion does not surpass 25\%. This observation suggests that even for leading technology firms, the majority of applications related to LLMs may fall outside their R\&D scope. Related to this observation, the theoretical analysis suggests that the incentives for open-sourcing advanced software related to a multi-purpose technology are strongest when a firm possesses an intermediate number of compatible applications.

\begin{table}[htb!]
    \centering
    \footnotesize
\caption{LLM Patent Applications by Top LLM Repository Owners}
\label{tab:app_repo}
\begin{threeparttable}
\begin{tabular}{lccccccc}
\toprule
\makecell{Company} & \makecell{LLM\\Repos} & \makecell{LLM\\Compatibility} & \makecell{Patent\\App.} & \makecell{NLG\\App.} & \makecell{NLG\\CPC} & \makecell{Share\\ CPC NLG} \\ 
\midrule
Microsoft & 199 & 0.73 & 7,752 & 34 & 141 & 0.12 \\
Google & 118 & 0.75 & 7,664 & 47 & 155 & 0.13 \\
Meta & 91 & 0.58 & 2,834 & 27 & 130 & 0.11 \\
Alibaba & 56 & 0.60 & 2,242 & 3 & 14 & 0.01 \\
Salesforce & 48 & 0.75 & 1,834 & 17 & 61 & 0.05 \\
IBM & 36 & 0.76 & 19,142 & 115 & 301 & 0.25 \\
Amazon & 30 & 0.67 & 1,840 & 4 & 15 & 0.01 \\
Nvidia & 18 & 0.70 & 1,851 & 1 & 5 & 0.00 \\
Intel & 16 & 0.52 & 11,090 & 4 & 22 & 0.02 \\
Tencent & 13 & 0.57 & 4,486 & 15 & 97 & 0.08 \\ 
\bottomrule
\end{tabular}

\begin{tablenotes}[para,flushleft]
\emph{Notes:}  The table presents selected statistics for 10 firms with the largest number of repositories of LLM-related papers (LLM Repos) on GitHub. `Transformer Sim.' denotes the cosine similarity with patents citing the Transformer paper \citep{vaswani2017attention}. `Pat. App.' refers to the number of patent applications filed by that applicant within the dataset. `NLG App.' indicates the number of applications with the CPC code G06F40/56, which denotes natural language generation. `NLG CPC' represents the total number of unique CPC codes co-occurring in natural language generation patents, and `Share CPC NLG' quantifies the firm's share of all such CPC codes.
\end{tablenotes}
\end{threeparttable}
\end{table}


\section{Empirical Analysis}
\label{sec:empirics}

\subsection{Model Quality and Open-Sourcing Decisions}
I start this section by examining how the quality advantage of LLMs over leading open source alternatives influences the developers' open source decisions. The analysis uses models from the Ecosystem dataset, provided by the Center for Research on Foundation Models (CRFM) at Stanford University, that have MMLU scores available. The MMLU is a widely-used benchmark to assess the general performance of LLMs\footnote{The Massive Multitask Language Understanding (MMLU) benchmark comprises a diverse set of natural language understanding tasks, assessing a model’s proficiency across various subjects and question types. The random guess baseline score is 25. The scores were sourced from the LMSYS Chatbot Arena Leaderboard \citep{chiang2024chatbot}, Paperswithcode.com, the Huggingace Open LLM Leaderboard, and the models' release reports; they reflect the "5-shot" performance of the models on the benchmark. }. Figure \ref{fig:mmlu} illustrates how the leading LLMs' performance on this benchmark has evolved over time.

 As displayed in Figure\ref{fig:mmlu}, there has been a persistent gap in performance quality between proprietary and open source models. Early LLMs, such as OpenAI's GPT-2, were primarily open sourced and used for research purposes. By contemporary standards, these early models had limited capabilities. GPT-3, a pioneering proprietary LLM, was significantly more advanced than its open source counterparts at the time. OpenAI's decision to adopt a proprietary release strategy for GPT-3 is aligned with the predictions of the theoretical framework in Section \ref{sec:model}, suggesting that a LLM developer will adopt a proprietary release strategy if the model's lead over its open source alternative is large enough. Since the release of GPT-3, closed LLMs have stayed ahead of the curve. Nevertheless,  high-quality open source LLMs have narrowed this gap between open and closed models\footnote{For domain specific tasks such as coding or mathematical reasoning, fine-tuned models have already achieved comparable performance to the state-of-the-art closed models For example, GPT-4 surpasses LLaMA and LLaMA-2  by wide margins on GSM8K (mathematical reasoning) and HumanEval (code generation) benchmarks. However, two LLaMA-based models, MathCoder and WizardCoder, respectively, approach or even slightly exceed GPT-4's performance on GSM8K and HumanEval, according to the models' documentation available at the time of their release.}.

 Further suggestive evidence worth noting concerns the developers of frontier models. Most top-tier closed LLMs have been released by a few organizations, particularly Google, OpenAI, and Anthropic. Considering the scaling laws of LLMs \citep{kaplan2020scaling}, a model's performance is primarily determined by its size, training data, and computational resources. Therefore, training LLMs that can outperform previous state-of-the-art models tends to be increasingly costly and out of reach for organizations without substantial resources. Nevertheless, the open source ecosystem has shown greater dynamism in releasing models that surpass previous frontiers. This greater dynamism can partly be explained by the fact that, contrary to the closed paradigm, in the open source ecosystem developers can build on each other's efforts, leading to more frequent breakthroughs in state-of-the-art models.
\begin{figure}[ht!]
\centering
\footnotesize
\caption{Quality Evolution of Frontier Open and Closed LLMs}
\label{fig:mmlu}
    \centering
        \includegraphics[width=.8\linewidth]{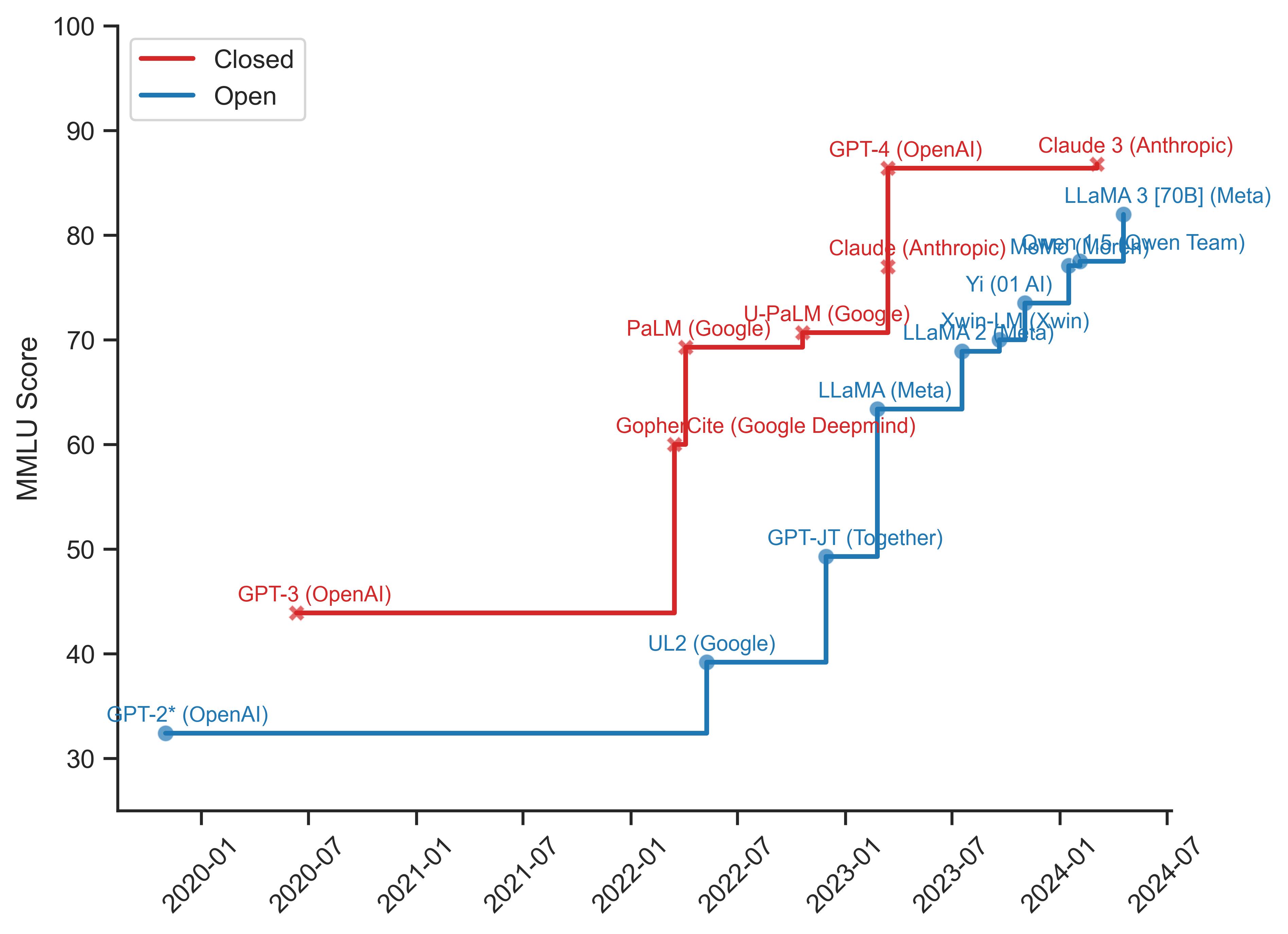}
        
    \hspace{0.4cm}\parbox{\textwidth}{\footnotesize{\textit{Notes:} The figure depicts the evolution of the performance of open and closed frontier LLMs in the CRFM data as measured by the Massive Multitask Language Understanding (MMLU) benchmark. A frontier open (closed) model is defined as a model that outperforms its preceding open (closed) models on this specific benchmark. The name of the developers are provided in parentheses. *The scores for GPT-2 reflects the score of the fine-tuned model.  }}
\end{figure}

To further examine the relationship between model quality and open-sourcing decision, consider the following regression,
 \begin{equation}
 \label{eq:reg_os}
      y_{i} = \alpha + \beta \,  \left(Q_i - Q^{*}_{O,t} \right) + \gamma X_{i} + \varepsilon_{i}
 \end{equation}
where $y_i$ is a binary outcome equal to one if model $i$ is open sourced. The main independent variable of the regression is $ \left(Q_i - Q^{*}_{O,t} \right)$ that shows the difference between the quality of model $i$ and the quality of the best available open source model. The theoretical framework predicts that $\beta$ is negative. That is, ceteris paribus, if a model surpasses its existing open source alternative by a wider margin, the owner is less likely to open source it. 

The main challenge for estimating the above regression is that there is no universally available and agreed upon measure of quality for LLMs. Even widely-used benchmarks like MMLU are available for only a subset of models in the CRFM dataset, where score availability is likely influenced by model quality. Nevertheless, if a model's performance is inferior to that of a comparable top-tier open source model, marginal quality improvements are unlikely to influence the owner’s decision to open source. Hence, my focus is on top-tier models, where benchmark score data are more readily available and the relationship between model quality and open-sourcing decisions is most relevant.

Table \ref{tab:reg_os_decision} presents the linear-probability-model (LPM) estimates of the parameter of interest $\beta$. As expected, all estimates of $\beta$ have a negative sign, suggesting that a larger gap between a model and the best available open source option decreases the likelihood that the model will be open sourced. Furthermore, the estimates of $\beta$ suggests economically significant correlations. A 10-point (out of 100) increase in the performance of the model on MMLU with respect to the best available open source option is associated with a 10-11 percentage point decrease in the likelihood of being open sourced.  The estimates of $\beta$ change only marginally after including the level of reported (predicted) model quality. The estimates of the level variable are statistically indistinguishable from zero and economically negligible, indicating that the level of model quality is only weakly correlated with the decision to open source, once the quality difference between the model and the leading open source alternative is considered.

\begin{table}[htp]
    \centering
    \footnotesize
\caption{Model Quality Lead and Open Sourcing Decision}
\label{tab:reg_os_decision}
\begin{threeparttable}
    \begin{tabular}{lcccc} \toprule
     & (1) & (2) & (3) & (4) \\ \midrule \addlinespace
    $Q - Q^{*}_{o,t}$ & -0.011*** & -0.011*** & -0.010*** & -0.010*** \\
 & (0.002) & (0.003) & (0.003) & (0.003) \\ \addlinespace
    $Q$ &  & -0.000 & 0.000 & 0.000 \\
 &  & (0.003) & (0.003) & (0.003) \\ \addlinespace
    \textit{For-Profit} &  &  & -0.138*** & -0.178*** \\
 &  &  & (0.052) & (0.060) \\ \addlinespace
    \textit{Big-Tech} &  &  &  & 0.210** \\
 &  &  &  & (0.094) \\ \addlinespace
     $R^2$ & 0.312 & 0.312 & 0.339 & 0.372\\ 
     $N$ & 86 & 86 & 86 & 86 \\ \bottomrule
    \end{tabular}
\begin{tablenotes}[para,flushleft]
\emph{Notes:} The table presents the linear probability model regression estimates of relationship between model quality and open-sourcing decision. $Q$ is the reported (estimated) quality of the model measured by reported (estimated) performance on the MMLU benchmark. $Q^{*}_{o,t}$ is the quality of the state-of-the-art open sourced model at the time of model's release. For the first open sourced model, the state-of-the-art is considered to be random guess baseline of 25. \textit{For-Profit} is a dummy variable for a for-profit developer. \textit{Big-Tech} is a dummy variable showing if the model is released by one of the following corporations:  Google,  Meta, and Microsoft. Heteroskedasticity robust standard errors are displayed in parentheses. 
\end{tablenotes}
\end{threeparttable}
\end{table}

Unsurprisingly, the coefficients for \textit{For-Profit} organizations' dummy variables in columns (3-4) are negative and statistically significant, indicating that for-profit organizations are, on average, less likely to open source their models\footnote{The status of each organization was determined manually using online sources such as firms' websites and Crunchbase. Model size—defined as the number of parameters—is documented for all open models and most closed models.}. Conversely, the coefficient for the \textit{Big-Tech} dummy variable is positive, large, and statistically significant\footnote{\textit{Big-Tech} indicates if a model is released by Google, Meta or Microsoft. Other recognized Big Tech companies do not have a model included in the dataset.}. This finding is aligned with theoretical results, predicting that, ceteris paribus, Big Tech companies are more inclined to open source their models as they have more compatible applications that can benefit from the positive spillovers of the open source community.

\subsection{Open Source as an R\&D Catalyst}
\label{sec:llama}
On February 24, 2023, Meta introduced its large language model, named LLaMA, and made the model available to researchers in academia, industry, government, and civil organizations \citep{Meta2023LLaMA}. This section studies the influence of LLaMA on the activities of LLM researchers. Using activity on GitHub as a proxy for LLM researchers' efforts, I document a significant increase in research-related activities following the release of LLaMA. I must acknowledge the difficulty in making causal claims due to the active period of LLM research around the time of LLaMA's release and the absence of direct data on researchers' use of LLaMA. Nevertheless, the main finding is highly consistent across various specifications, estimators, time horizons, and methods of identifying LLM researchers. It also withstands multiple falsification checks, suggesting a potentially causal interpretation.

A key aspect of open-sourcing LLaMA was Meta's decision to not only provide the fine-tuned assistant model but also the code and parameters of the pre-trained model\footnote{LLaMA-1 was released under a non-commercial license. However, Meta later revised it stance leading to the release of LLaMA-2 in July 2023 under a more permissive commercial license.} (see Section \ref{sec:context}). This enabled researchers to leverage a high-quality pre-trained LLM to tailor it to specific applications,  potentially stimulating further research activities around LLMs. Consequently, LLaMA was widely adopted and served as a foundation for a subsequent generation of LLMs\footnote{Meta reported that LLaMA-based models have been downloaded over 30 million times through Hugging Face, and adopted by thousands of startups, innovators, and developers \citep{MetaReportLLaMA}.}. However, it remains unclear whether open-sourcing LLaMA merely replaced prior generations of open language models or stimulated further research activity among LLM researchers. This distinction is particularly important as replacing inferior models is equivalent to a one-time upward shift in LLM technology level. However, if open-sourcing LLaMA had a positive influence on R\&D activities around LLMs, it could amplify the growth rate of the technology beyond having a positive impact on its level.  
\subsection*{Data}
Given the notable surge in research on LLMs in 2023, a high-frequency measure of activity is essential to isolate the impact of specific events. Therefore, traditional metrics like the numbers of patent applications or publications, which are recorded with significant delays, are not appropriate proxies in this scenario. Consequently, this study utilizes weekly counts of \emph{contributions} on GitHub as an indicator of research activity\footnote{While users can opt to conceal this information, such instances are rare.}. GitHub defines several activities as contributions, with the primary method being code modifications in a repository (commit). Other forms include code reviews, issue management, and pull requests\footnote{Pull requests propose incorporating changes from one branch to another, usually for code review and integration before merging.}. GitHub also provides information on commits to public repositories, offering a more accurate measure of contributions to open research efforts.

I identified contributors to repositories of LLM-related papers as LLM researchers. To establish a control group, I selected a subset of GitHub users presumably unaffected by the open-sourcing of LLaMA. Given LLMs' significant impact on the broader AI research community, using AI researchers without LLM-related papers for the control group was deemed implausible. Therefore, I identified 20 major repositories on GitHub not directly related to AI, and subsequently collected information of all contributors to those repositories to serve as the control units\footnote{ Each chosen repository ranks as the most popular under a specific Topic on GitHub, determined by the number of \emph{Stars}. Repositories solely providing educational materials were omitted. Description of these topics and their corresponding repositories is provided in Appendix \ref{sec:app_data}}. While it is not possible to verify that open-sourcing LLaMA had no influence on the activities of this group of GitHub users, it is hard to think of a possible scenario in which open-sourcing LLaMA had negative impact on the activity among users in the control group\footnote{The raw data shows a slight increase in the average number contributions among such users in the post period.}. As a result, to the extent that open-sourcing LLaMA had a positive effect on the overall activity of users in this group, the results would underestimate the true influence of LLaMA on the activities of LLM researchers.

Similar to other platforms, users on GitHub often include biographical details on their profile pages. These largely unstructured biographies typically feature information about their location, as well as affiliations with universities, companies, or organizations. Given the impracticality of manually inspecting the vast number of profiles and the lack of structured data for accurate pattern-based processing, I employed an LLM for data parsing. Specifically, the LLM was tasked with extracting users' countries, their current sector of employment (Academia or Industry), and the names of affiliated organizations, provided this information was available. A manually inspected sample confirmed the LLM's qualitative performance. Further details on utilizing the LLM for data processing can be found in Appendix \ref{subsec:parse_llm}. In total, the profiles of over 63 thousand AI researchers were analyzed. The models' prediction suggested that ca. 53\% of these researchers work in Academia, 24\% work in Industry, and 22\% did not disclose this information. 

\subsection*{Results}

Consider the following event-study regression,
 \begin{equation}
 \label{eq:event}
      y_{i,t} = \alpha_i + \tau_t + \beta_{i,t} \, Treated_{i,t} + \varepsilon_{i,t}
 \end{equation}
where $y$ represents the relative deviation from the mean pre-event contribution  for user $i$ at time $t$, specifically, $(c_{i,t} - \Bar{c}_{i,pre})/\Bar{c}_{i,pre}$, where $c_{i,t}$ is user $i$ contributions at time $t$. This transformation allows interpreting the treatment effect as the average activity level change among affected researchers.  Using the raw counts of contributions yields a less intuitive interpretation and skew the results toward the highly active users. Raw contribution counts, while less intuitive and biased towards highly active users, do not alter the robustness of the findings when used as the outcome variable. The analysis includes LLM researchers with at least one pre-period contribution as treated units and non-AI repository contributors as controls. Researchers employed by Meta were excluded. The dataset comprises approximately 6,400 treated and 4,800 control group individuals, respectively. 

The goal here is to demonstrate that open-sourcing can stimulate research activities, rather than quantifying the precise impact of LLaMA on the GitHub contributions of LLM researchers. Contributions on GitHub serve primarily as a proxy for research activity. Therefore, even if a precise treatment effect of LLaMA's release on GitHub contributions could be estimated, its significance would be limited. Additionally, limitations in the data and identification strategy prevent strong causal claims regarding the treatment effect. Despite these limitations, the findings suggest that open-sourcing an advanced model can do more than replace inferior models; it may catalyze research activity within the community, potentially leading to further advancements and a snowball effect that accelerates technological growth in the field.

Figure \ref{fig:event} presents the estimates from the aforementioned event-study regression. The analysis spans a 21-week interval, with Week 0 defined as the seven days following LLaMA's release on February 24. Notably, the results reveal a moderate pre-trend in the activities of LLM researchers, in comparison to the control group. However, immediately after the model's release in Week 0, a significant decline in GitHub activities among LLM researchers is observed. This pattern is likely attributable to researchers allocating time to explore the new model rather than contributing to their existing projects. Subsequent to Week 0, LLM researchers' contributions exhibit an upward trend, stabilizing several weeks later.
\begin{figure}[ht]
\centering
\footnotesize
\caption{Impact of LLaMA on Contributions of LLM Researchers}
\label{fig:event}
    \centering
    \includegraphics[width=0.8\linewidth]{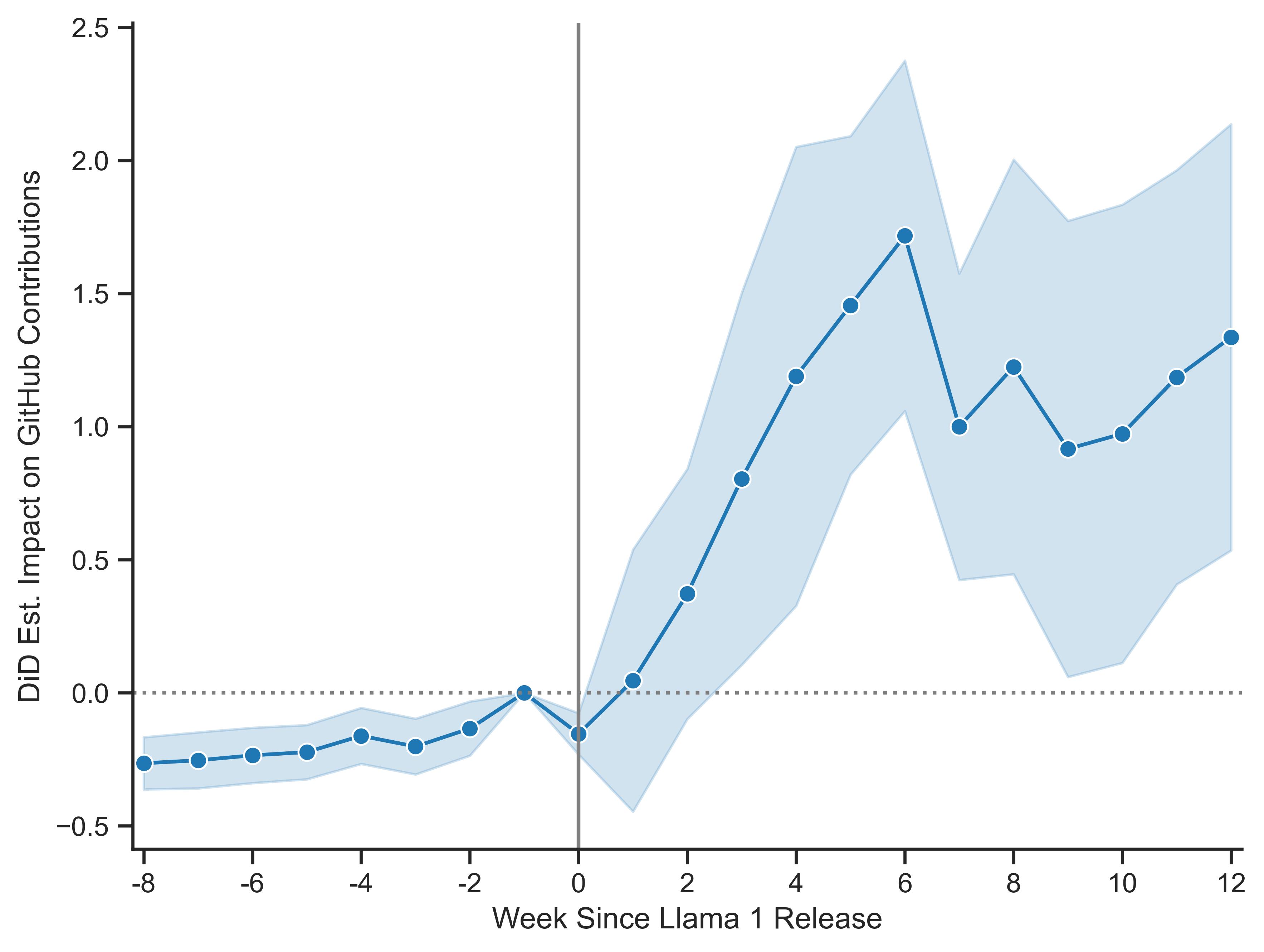}
    \hspace{0.4cm}\parbox{\textwidth}{\footnotesize{\textit{Notes:} The figure plots the coefficients from the event-study regression, as described in equation \ref{eq:event}. The dependent variable is the relative deviation of contributions from their mean pre-event level, defined as $y_{it} = (c_{i,t} - \Bar{c}{i,pre})/\Bar{c}{i,pre}$. The vertical line marks the introduction date of LLaMA. The shaded areas denote 95 percent confidence intervals, calculated based on standard errors that are clustered at the individual level.
    }}
    
\end{figure}

Table \ref{tab:did_contrs} presents the Difference-in-Differences (DiD) estimates of the impact of LLaMA's open-sourcing on GitHub contributions by LLM researchers (outcomes in Week 0 have been omitted from the analysis). The estimates reveal a substantial and statistically significant increase in the contributions of LLM researchers after LLaMA's release. This increase is observed for both academia and industry researchers. Including group-specific linear trends has negligible effects on the results, suggesting that the estimates are not primarily driven by pre-trends. The robustness analysis employs the Synthetic DiD estimator proposed by \cite{arkhangelsky2021synthetic} to account for complex pre-trend patterns. The findings are closely aligned with the baseline estimates. 

Furthermore, to ensure robustness, the analysis was narrowed to a 30-day period before and a 17-day period after the model's release and the results were produced using daily contributions data. This was to  done to isolate the estimates from the potential effects of GPT-4's release in March 14. The findings continue to indicate a sizable and significant increase in the contributions of LLM researchers after LLaMA's release, albeit smaller than the baseline figures. Such a reduced short-term effect was anticipated, as event-study estimates highlighted an increasing momentum in LLM researchers' activity levels post-LLaMA's release. Overall, estimates suggest a 40-140\% increase in LLM researchers' contributions after LLaMA's release, varying with the time horizon and estimation methodology.

\begin{table}[htp!]
    \centering
    \footnotesize
\caption{Impact of LLaMA on Activity of LLM Researcher on GitHub}
\label{tab:did_contrs}
\begin{threeparttable}
    \begin{tabular}{lcccccc}
    \toprule
     & (1) & (2) & (3) & (4) & (5) & (6) \\ 
     & All & All & Academy & Academy & Industry & Industry \\ \midrule
    ATT & 1.202*** & 1.338*** & 1.352*** & 1.453*** & 0.889** & 1.127* \\
        & (0.196) & (0.238) & (0.219) & (0.269) & (0.401) & (0.651) \\
     \addlinespace
    Obs. & 224,240 & 224,240 & 166,520 & 166,520 & 128,000 & 128,000 \\
    $R^2$ & 0.006 & 0.006 & 0.007 & 0.007 & 0.004 & 0.004 \\
    N. Ind. & 11,212 & 11,212 & 8,326 & 8,326 & 6,400 & 6,400 \\
    Ind. FE & Y & Y & Y & Y & Y & Y \\
    Time FE & Y & N & Y & N & Y & N \\
    Linear Trend & N & Y & N & Y & N & Y \\
    Trend x Treat & N & Y & N & Y & N & Y \\ \bottomrule
    \end{tabular}
\begin{tablenotes}[para,flushleft]
\emph{Notes:} The table presents the Difference-in-Differences estimates of the impact of LLaMA on the activity of LLM researchers on GitHub. The dependent variable is the relative deviation of weekly contributions from their mean pre-event level. The outcomes for Week 0 (the first seven days after LLaMA's announcement) are omitted. `Academy' indicates the group of LLM researchers whose GitHub profiles indicate that they are working in academia, and `Industry' indicates the estimates for LLM researchers whose GitHub profiles indicate they are employed in the industry. Cluster-robust standard errors are displayed in parentheses. 
\end{tablenotes}
\end{threeparttable}
\end{table}

Table \ref{tab:reg_counts} presents the TWFE and Poisson regression estimates of changes in contributions by LLM researchers on GitHub, as measured by raw counts of contributions. The results from both estimators are aligned with the baseline results, implying that open-sourcing LLaMA made a positive influence on LLM researchers' activities on GitHub.

\begin{table}[ht]
\centering
\footnotesize
\caption{Impact of open source on GitHub Contributions}
\label{tab:reg_counts}
\begin{threeparttable}
    \begin{tabular}{lcccccc}
    \toprule
    & (1) & (2) & (3) & (4) & (5) & (6) \\
    & All & All & Academy & Academy & Industry & Industry\\
    & TWFE & Poisson & TWFE & Poisson & TWFE & Poisson \\ \midrule
    ATT & 0.780*** & 1.133*** & 0.801*** & 1.168*** & 0.795** & 1.102*** \\
        & (0.270) & (0.0261) & (0.292) & (0.0367) & (0.404) & (0.0402) \\
    \addlinespace
    Obs & 224,240 & 224,080 & 166,520 & 166,420 & 128,000 & 127,880 \\
    $R^2$ & 0.004 &  & 0.004 &  & 0.004 &  \\
    N. Ind. & 11,212 & 11,204 & 8,326 & 8,321 & 6,400 & 6,394 \\
    Ind. FE & Y & Y & Y & Y & Y & Y \\
    Time FE & Y & Y & Y & Y & Y & Y \\
    Trend & N & N & N & N & N & N \\
    Trend x Treat & N & N & N & N & N & N \\
    \bottomrule
    \end{tabular}
\begin{tablenotes}[para,flushleft]
\emph{Notes:} The table presents the Difference-in-Differences estimates of the impact of LLaMA on the total weekly contributions of LLM researchers on GitHub. The dependent variable is total weekly contributions. TWFE denotes Two-way fixed-effects estimator, and Poisson denotes the Poisson regression with conditional fixed-effects. The outcomes for Week 0 (the first seven days after LLaMA's announcement) are omitted. Academy' indicates the group of LLM researchers whose GitHub profiles indicate that they are working in academia, and Industry' indicates the estimates for LLM researchers whose GitHub profiles indicate they are employed in the industry. Cluster-robust standard errors are displayed in parentheses.
\end{tablenotes}
\end{threeparttable}
\end{table}

Table \ref{tab:did_commits} presents the Difference-in-Differences (DiD) estimates for changes in commit activity to public repositories by LLM researchers. The findings show a significant increase in public contributions on GitHub by LLM researchers following the release of LLaMA. Further analysis indicates that this increase was primarily in the academic sector, with no notable change in public commit activity among industry-based LLM researchers.

\begin{table}[htp]
    \centering
    \footnotesize
\caption{Impact of open source on Public Contributions}
\label{tab:did_commits}
\begin{threeparttable}
    \begin{tabular}{lccc}
    \toprule
     & (1) & (2) & (3) \\ 
     & All & Academy & Industry \\ \midrule
    ATT & 0.372*** & 0.428*** & 0.274 \\
     & (0.133) & (0.145) & (0.199) \\
     \addlinespace
     Obs. & 224,240 & 166,520 & 128,000 \\
    $R^2$ & 0.001 & 0.001 & 0.001 \\
    N. Ind.& 11,212 & 8,326 & 6,400 \\
    Ind. FE & Y & Y & Y \\
    Time FE & Y & Y & Y \\ \bottomrule
    \end{tabular}
\begin{tablenotes}[para,flushleft]
\emph{Notes:} The table presents the Difference-in-Differences estimates of the impact of LLaMA on weekly public contributions of LLM researchers on GitHub. The outcomes for Week 0 (the first seven days after LLaMA's announcement) are omitted. `Academy' indicates the group of LLM researchers whose GitHub profiles indicate that they are working in academia, and `Industry' indicates the estimates for LLM researchers whose GitHub profiles indicate they are employed in the industry. Cluster-robust standard errors are displayed in parentheses. 
\end{tablenotes}
\end{threeparttable}
\end{table}

\subsection*{Robustness Analysis}

\subsubsection*{Non Linear Pre-Trends}
Table \ref{tab:did_contrs} rules out the possibility that linear pre-trends drive the estimates. Table \ref{tab:reg_sdid} presents Synthetic DiD estimates of the treatment effects, offering flexible control over higher-order pre-trends in outcomes between treated and non-treated units. The results align closely with the baseline DiD estimates.

\subsubsection*{Simultaneous Shocks}
The first half of 2023 witnessed significant research activity in LLMs. DiD estimates are susceptible to other sources of shocks potentially influencing LLM researchers around the time of LLaMA-1 being open sourced. Notably, GPT-4, the successor to GPT-3.5 (ChatGPT), was released in mid-March 2023. Despite being an enhanced version of its closed predecessor, GPT-4 introduced additional features, such as multimodality\footnote{Multimodality refers to the ability of a single model to process and generate multiple types of data.} and advanced code generation capabilities, which could impact research activities. To isolate the results from GPT-4's potential influence, I use daily data and narrow the time window to 30 days before LLaMA's release and one day prior to GPT-4's release. The short-term estimates, presented in Table \ref{tab:reg_daily},  remain statistically and economically significant, but are smaller in size compared to the baseline values,  likely due to the impact of open-sourcing LLaMA not being fully realized in the few weeks after its release.

In addition, Figure \ref{fig:event_daily} plots the estimated event-study coefficients of Equation \ref{eq:event} using daily data around the day of open-sourcing LLaMA. The figure does not reveal any significant increase in the contributions of LLM researchers following the release of GPT-4, and the overall pattern suggests that the trend in contributions remains relatively consistent after open-sourcing LLaMA. Overall, the robustness analysis indicates that the release of GPT-4 is unlikely to account for the main findings.

\begin{figure}[ht]
\centering
\footnotesize
\caption{Short-Term Impact of LLaMA on Contributions of LLM Researchers
\label{fig:event_daily}}
    \centering
    \includegraphics[width=0.75\linewidth]{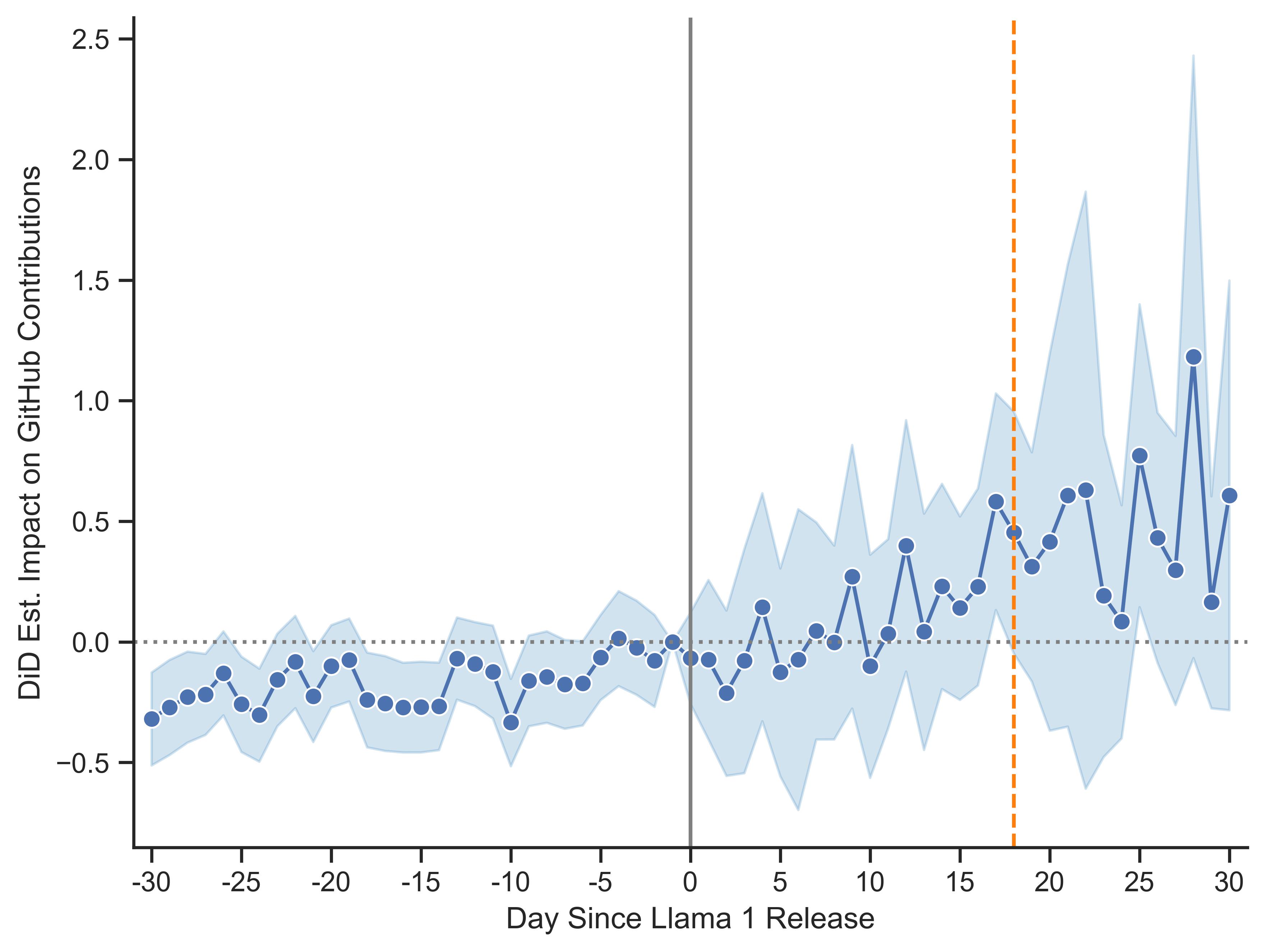}
    \hspace{0.4cm}\parbox{\textwidth}{\footnotesize{\textit{Notes:} The figure plots the coefficients from the event-study regression, as described in Equation \ref{eq:event}. The dependent variable is the daily relative deviation of contributions from their mean pre-event level, defined as $y_{it} = (c_{i,t} - \Bar{c}{i,pre})/\Bar{c}{i,pre}$. The solid vertical line (in grey) marks the introduction date of LLaMA. The dashed vertical line (in orange) denotes the release of GPT-4. The shaded areas represent 95 percent confidence intervals, based on clustered robust standard errors.
    }}

\end{figure}

Submission deadlines of major AI conferences could potentially be another source of shocks. However, such shocks are unlikely to explain the steep rise in the number of contributions after open-sourcing LLaMA and its stabilization after passing 5-6 weeks, as displayed in Figure \ref{fig:event}. Nevertheless, data on the submission deadlines of two major NLP conferences (ACL and EMNLP) and two major general AI conferences (ICML and NeurIPS) were collected. The submission deadlines for ACL and ICML fell in late January, before the open-sourcing of LLaMA, while EMNLP's deadline was in late June, a few weeks after the end of the study period. The submission deadline for NeurIPS on May 17, coinciding with Week 11 in Figure \ref{fig:event}, appears to be the only potential concern. However, given that the estimate for Week 12 appears to be as large as the estimate for Week 11, it seems unlikely that this deadline had a meaningful impact on the main findings.
 
\subsubsection*{Alternative Categorization of the Treated Group}
Table \ref{tab:reg_treat} demonstrates the robustness of the results to alternative methodologies for identifying researchers in the treated group. First, papers utilizing the term `language model' in their titles or abstracts are identified, and the contributors to the repositories associated with these papers are subsequently recognized as the treated group. In the second strategy, a text-based clustering method is employed to discern the largest cluster associated with natural language processing\footnote{Specifically, K-means clustering is applied to the combined titles and abstracts of papers after dimensionality reduction (utilizing TSVD+UMAP) on the TF-IDF matrix of pre-processed texts. The most frequent cluster among the papers presented at two conferences dedicated to Computational Linguistics (ACL) and Natural Language Processing (EMNLP) is then identified.}. Contributors to repositories of papers in the cluster linked to NLP are then determined as the treated group. The results under both methodologies align closely with the baseline estimates.

\section{Theoretical Analysis}
\label{sec:model}
This section introduces a model that explores the dynamics of AI software development and the decision-making process regarding open sourcing. It concentrates on a handful of factors considered to have primary importance, which can be readily incorporated into simple growth models. The goal is to develop a straightforward dynamic discrete choice model to analyze the decisions of a tech firm concerning the open-sourcing a LLM, which serves as an input for producing software applications.

\subsection{Environment}
\subsubsection*{LLMs as a GPT}
The model positions LLMs as an enabling technology that enhances profits in the downstream software application sector (AS). The AS comprises a unit mass of software producers, each differing in LLM-compatibility.  Profits from LLM-integration in the AS are contingent on the language model quality, and the amount of \emph{compute}\footnote{Compute denotes the total computing resources utilized for algorithm training and execution.} used to integrate the model with that application. Furthermore, using compute improves the model's quality for the subsequent periods.

In the AS, software producers are uniformly distributed across the unit interval ($x_i \sim U[0,1]$), and in each period, they choose to use one type of LLM with quality $q_{\tau, t}$.  The profit $\pi_{i,\tau,t}$ for producer $i$, at time $t$, spending $k_{i,t}$ units of compute when using model $\tau$  with quality $q_{\tau,t}$ is: 

\begin{equation}
    \pi_{i,t, \tau} = e^{-\gamma x_i} \left(q_{\tau,t} k_{i,t} \right)^\alpha - k_{i,t} - P_{\tau,t}
\label{eq:profit}
\end{equation}
where $e^{-\gamma x_i}$ represents the producer's compatibility with LLMs and $P_{\tau,t}$ is the license price of type $\tau$ model in period $t$. 

A LLM is either proprietary/closed or open source. An open source model is freely accessible to all producers, while proprietary model is priced by its owner at the beginning of each period. Moreover, while only the owner can improve quality of the closed model, the open source model benefits from collaborative enhancements by all software producers. There are two types $\tau$ of models, $A$ and $B$. Type $B$ is open sourced with $P_{B,t}=0$ for all $t$, while type $A$ is developed and owned by tech firm $\mathbf{A}$.
\subsubsection*{The Tech Firm}
The model includes a (big) tech firm $\mathbf{A}$, which owns and controls all software producers in the application sector within the range $x_i \in [0, m]$. While subsequent analysis will explore $\mathbf{A}$'s decision in developing its own LLM, it is initially assumed that at $t=t_0$, $\mathbf{A}$ is endowed with model $A$ with superior quality compared to the open source alternative. From $t \geq t_0$ onwards, Firm $\mathbf{A}$ can choose to irreversibly open source its model, under the condition $q_{A,t} \geq q_{B,t}$\footnote{This assumption is without loss of generality, as open-sourcing lower-quality model would not influence external producers' choices.}.

Open-sourcing allows for external contributions, potentially accelerating the quality growth of model $A$ and, consequently, increasing $\mathbf{A}$'s profits from internal producers. Alternatively, Firm $\mathbf{A}$ can maintain a closed API, licensing the model to external producers for direct revenue, but this approach excludes the possibility of external quality improvements from the open source community. The choice between rapid quality enhancement and direct  monetization presents a strategic dilemma in open-sourcing decisions. 

\subsubsection*{External Producers}
Producers with $x_i \in (m,1]$ that are not controlled by firm $\mathbf{A}$ face a static profit-maximizing problem each period. They evaluate the qualities and API  prices of all available LLMs types and choose the model that maximizes their profits, based on their optimal compute unit for the selected model type. Specifically they solve,

\begin{equation*}
     \pi_{i,t, \tau} = \max_{\tau,k_\tau} \left\{ \left(e ^{-\gamma x_i} \left(q_{A,t} k_{i,A,t}\right)^\alpha - k_{i,A,t} - P_{A,t} \right), \left( e^{-\gamma x_i} \left(q_{B,t} k_{i,B,t} \right)^\alpha - k_{i,B,t} \right) \right\}
\end{equation*}
assuming free access to open source model $B$. If Firm $\mathbf{A}$ chooses to open source the higher-quality model $A$, the focus for external producers shifts to optimizing compute levels $k_{A,t}$ for $A$.

When $\mathbf{A}$ offers its model through an API, external producers compare the license fee $P_{t}$ (simplifying the subscript $A$ in $P_{A,t}$) with the quality disparities between $A$ and $B$ in their decision-making process. This results in the following demand relation for the API of model $A$:
\begin{equation}
    Q_{A,t} =  \frac{1}{\gamma}  \left[\ln\left(\alpha \left(q_{A,t}^{\alpha/(1-\alpha)} - q_{B,t}^{\alpha/(1-\alpha)}\right)^{1-\alpha}\right) - (1 - \alpha) \ln\left(\frac{\alpha P_t}{1 - \alpha}\right)\right] - m
    \label{eq:api_demand}
\end{equation}
The demand relation implies that the demand for model $A$ increases with its quality lead over model $B$, and decreases with its price $P$. Note that as Firm $\mathbf{A}$'s internal producers have an open access to the model, their mass, denoted by $m$, is subtracted from the demand relation. For a detailed derivation of the relationship, see Appendix \ref{sec:app_model}.

\subsection{Open Source Decision}

Firm $\mathbf{A}$'s decision is to either irreversibly open source model $A$ or maintain its closed status for another period. Consequently, the firm must weigh the value of open model, $V^O(q)$, against the value of closed model, $V^C(q,q_B)$, where $q$ denotes the quality of model $A$. Notably, the value of open source model, $V^O(q)$, does not depend on the quality of its open source alternative, as $q \geq q_B$ implies that every producer would opt for $A$ if it were open sourced. However, the value of the closed model, $V^C$, is influenced by $q_B$, since the demand for $A$'s API hinges on both the quality of $A$ and its open source counterpart.

This formulation of the firm's decision-making process aligns well with a dynamic programming model featuring a discrete choice. Consequently, the value of the open model is derived as a solution to the following Bellman equation:
\begin{equation}
\centering
\begin{gathered}
       V^{O}(q) = \max_{k(x)} \left[ \int_{0}^{m} \pi(q, k(x)) \, dx + \beta V^{o}(q') \right] \\
\text{s.t.} \quad q' = q + \psi K_A + \phi K_{-A} 
\end{gathered}
\label{eq:V_open}
\end{equation}
where $\beta$ is the time discount factor, $\psi K_A$ and $\phi K_{-A}$ reflect the contribution of internal and external computes to the model’s quality in the subsequent period, with $\phi, \psi \in (0,1]$. The integral on the right-hand side of the equation denotes the total profit generated by all software producers owned by Firm $\mathbf{A}$ in the application sector. Equation \ref{eq:V_open} implies that, in any period, when $k(x)$ is chosen optimally, $\mathbf{A}$’s valuation of open model with quality $q$ equals the immediate profit generated by its producers using the model, plus the discounted value of the enhanced model with improved quality $q'$.

The maximization problem on the right-hand side of Equation \ref{eq:V_open} can be simplified with respect to the optimal compute function $k(x)$. It is important to note that the transition equation depends only on the aggregate internally used compute $K_A$, and not on the distribution of compute among $\mathbf{A}$'s producers. Consequently, when $K_A$ is optimally chosen, profit maximization implies that marginal profit with respect to $k$ must be equal for any two producers $x_i, x_j \in [0,m]$ owned by $\mathbf{A}$. After some algebraic manipulation,  Firm $\mathbf{A}$'s aggregate production profit in terms of $K_A$, can be expressed as:
 \begin{equation*}
      \Pi^F(q,K_A) = \Theta \left( q K_A \right)^\alpha - K_A
 \end{equation*}
 where $\Theta$ is a constant term determined by the parameters $\alpha$, $\gamma$ and $m$. For further details on derivations, see Appendix \ref{sec:app_model}. 
 
Subsequently, Equation \ref{eq:V_open}  can be simplified as
\begin{equation}
\centering
\begin{gathered}
       V^{O}(q) = \max_{K_A} \left[ \Pi^F(q,K_A) + \beta V^{o}(q') \right] \\
\text{s.t.} \quad q' = q + \psi K_A + \phi K_{-A} 
\end{gathered}
\label{eq:V_open_simp}
\end{equation}

The value of the closed model $V^C(q, q_B)$ depends not only on $\Pi^F(q, K_A)$ but also on the profit from model $A$'s API, $\Pi^A(q,q_B,P) = P Q_{A,t}$, where $Q_{A,t}$ is the demand for the model's API as given by Equation \ref{eq:api_demand}. Therefore, the Bellman equation for the closed LLM is\footnote{The equation presents a simplified version of the firm's problem modeled in the numerical analysis. To analyze the model for the complete set of states for $q_A$ and $q_B$, the complete problem introduces another option which is possibility of switching from own model to alternative open source LLM after paying some adjustment cost, specifically when $q_A < q_B$. However, as the analysis concentrates on cases where $q_A \geq q_B$, this simplification does not impact the firm's decision regarding open sourcing in the deterministic case. Intuitively, it is not optimal for the firm to forfeit the opportunity to open source its higher-quality model and set the API price so high that it eventually necessitates incurring adjustment costs to switch to an alternative open source LLM. This intuition is also confirmed by the results of the numerical analysis of the model.}:
\begin{equation}
\centering
\begin{gathered}
       V^{C}(q, q_B) = \max_{K_A,P} \left[ \Pi^F(q,K_A) + \Pi^A(q,q_B,P) + \beta \max \left\{ V^{O}(q'), V^C(q', q_B')\right\}   \right] \\
\text{s.t.} \quad q' = q + \psi K_A \\
\text{\& }   \quad q'_B = q_B + \phi K_{B} 
\end{gathered}
\label{eq:V_closed}
\end{equation}

Equation \ref{eq:V_closed} implies that if $\mathbf{A}$ decides to keep its model closed, it gains profits from both production and its model's API, while retaining the option to open source or keep the model closed in the next period, hence obtaining the discounted value of $$\max\{ V^{O}(q'), V^C(q', q_B')\}$$ However, since the model remains closed, its quality in the next period $q'$ only increases due to internally used computes $K_A$. Additionally, the compute used by external producers using model $B$ enhances the next period's quality of $B$ by $\phi K_B$. Consequently, when setting the API price, Firm $\mathbf{A}$ must consider its impact on immediate profits and its future implications via the channel of $q_B$. An increase in $P$ not only alters $\Pi^A$ immediately but also affects future profits as producers switching from $A$'s API to the open source alternative $B$ will contribute to the growth of model $B$, thereby affecting the demand for $A$'s API in all subsequent periods.

Finally, Firm $\mathbf{A}$'s decision to open source can be modeled as choosing the option with the highest value:
\begin{equation*}
    V(q,q_b) = \max \left\{ V^O(q), V^C(q, q_B) \right\}
\end{equation*}
\subsubsection*{Analytical Findings}

Since the open-sourcing decision involves a discrete choice, an analytical solution does not exist\footnote{Dynamic programming models are typically analyzed numerically. Though analytical solutions can be obtained for some special cases using a guess and verify approach, the presence of a discrete choice in the model precludes an analytical solution due to the non-differentiability of the value function.}. However, it is still possible to characterize some key properties of the solution under mild assumptions.

The following proposition establishes that, keeping other variables including $q_B$ constant, if there exists a $q^*$ such that $V^C(q, q_B) = V^O(q)$, then $q^*$ acts as a critical threshold. That is, firm $\mathbf{A}$ will open source its model only if $q < q^*$ and will keep the model closed if $q > q^*$.

\begin{proposition} \label{prop:threshold}
Let \( q^* > q_{\mathbf{B}} \) be such that \( V^{\mathbf{C}}(q^*, q_{\mathbf{B}}) = V^{\mathbf{O}}(q^*) \). Then, the following conditions hold:
\begin{itemize}
    \item If \( q > q^* \), then \( V^{\mathbf{C}}(q, q_{\mathbf{B}}) > V^{\mathbf{O}}(q) \).
    \item If \( q < q^* \), then \( V^{\mathbf{C}}(q, q_{\mathbf{B}}) < V^{\mathbf{O}}(q) \).
\end{itemize}
\end{proposition}
\begin{proof}
The complete proof can be found in Appendix \ref{sec:app_model}. Briefly, the proof involves a first-order approximation of $V^C$  and $V^O$ around $q^*$, and exploiting that $V^C$ and $V^O$ are strictly increasing  with respect to $q$ and $V^C$ is strictly decreasing with respect to $q_B$.
\end{proof}

The second proposition implies that the firm sets the price of its model's API below its one-period revenue maximizing value to limit the growth of model $B$.
\begin{proposition} \label{prop:price}
Unless the optimal solution implies \( V^{\mathbf{C}}(q', q'_{\mathbf{B}}) < V^{\mathbf{O}}(q') \), the firm sets \( P \) below its one-period revenue-maximizing value.
\end{proposition}

\begin{proof}
    The proof involves using the first order condition and the envelop theorem with respect to $P$ and $q_B$. See the complete proof in Appendix \ref{sec:app_model}. \
\end{proof}
\subsubsection*{Numerical Analysis}
Dynamic programming models that involve a discrete choice must be analyzed numerically. For this purpose, I use Value Function Iteration (VFI) method, one of the most common methods to solve dynamic programming models in economics. VFI hinges on the assumption that the transformation of the value function is a contraction. This technique begins with an arbitrary initial guess for the value function, which is then iteratively updated using the Bellman equation and guaranteed to converge under a couple of weak and commonly made assumptions. For additional details on the numerical analysis, please refer to Appendix \ref{sec:app_model}.

Table \ref{tab:model_params} presents the baseline choice of parameters used in the numerical analysis. Note that the choices of these parameters are not to be taken literally. The goal here is to see how the tendency to open source generally varies with changes in a few key factors of interest within a simple framework, but the numerical values are not important.
\begin{table}[h] 
\footnotesize
    \centering
    \caption{Baseline Choices of Parameters}
\begin{threeparttable}   
\begin{tabular}{lcc}
\toprule
Description & Notation & Baseline Value  \\
\midrule
Time Discount Factor & $\beta$ & 0.9  \\
AI Compatibility Parameter &  $\gamma$ & 1.0 \\
Shape Production Function & $\alpha$ & 0.45 \\
Size Firm $\mathbf{A}$ & $m$ & 0.2 \\
Efficiency of Internal Development & $\psi$ & 0.5 \\
Efficiency of open source Ecosystem & $\phi$ & 0.5 \\
LLM Development Improvement Factor & $\lambda$ & 5.0 \\
LLM Development Cost Factor & $c_D$ & 0.4 \\
\bottomrule
\end{tabular}
\end{threeparttable}
\label{tab:model_params}
\end{table}

\subsection{Results}
In the following section, I present the results from the numerical analysis of the model.

\subsubsection*{Open Sourcing Decision}
Figure \ref{fig:model_decision} illustrates how firm $\mathbf{A}$'s valuation of the open and closed model varies with model quality, for a given quality of the alternative model $B$. Where model $A$ has a modest lead over model $B$, the value of closed model is smaller than the value of open source model. In this region, Firm $\mathbf{A}$ will open source its model. However, there is a critical threshold such that the value of open and closed model intersect. This critical threshold, which is marked by the vertical dashed line in the figure, specifies the open source window, implying that Firm $\mathbf{A}$ will open source $A$ only if $q_A$ is in this window, and for the qualities above that critical threshold $\mathbf{A}$ will keep its model closed.  Furthermore, Figure \ref{fig:os_window_q_B}  demonstrates how the size of the open source window changes with respect to $q_B$. The results show that while the absolute size of the window increases for larger values of $q_B$, its relative size remains roughly constant.
\begin{figure}[htp]
\centering
\footnotesize
\caption{Open Sourcing Decision}
\label{fig:model_decision}
    \centering
    \includegraphics[width=0.75\linewidth]{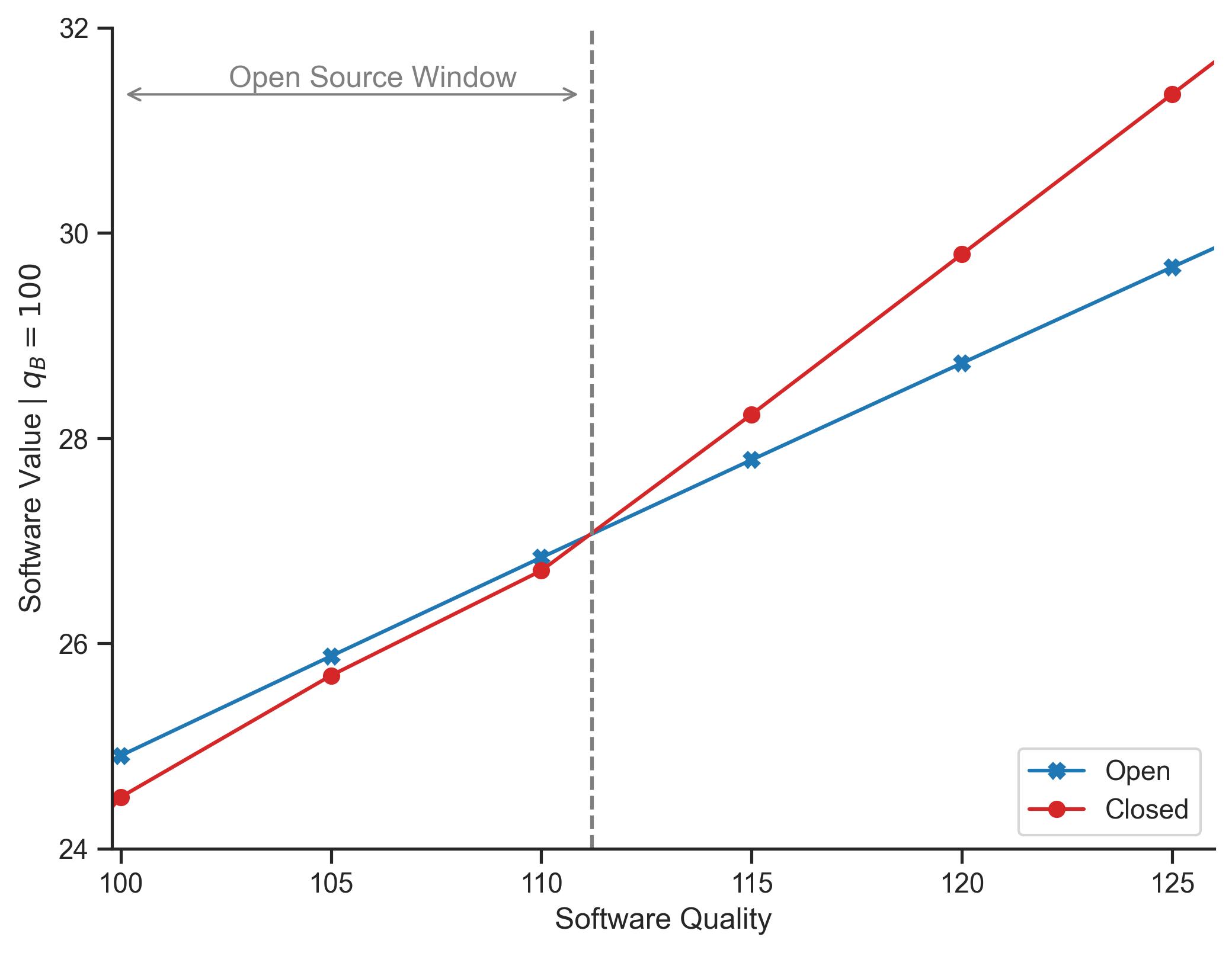}
    
    \parbox{\textwidth}{\footnotesize{\textit{Notes:} The figure depicts the valuation of closed and open model for firm $\mathbf{A}$ as a function of its quality, while maintaining the quality of model $B$ at a constant value of 100. The modeling parameters used in the figure are detailed in Table \ref{tab:model_params}.}}
\end{figure}
\subsubsection*{Efficiency of the open source Community}

The efficiency of the open source community in contributing to model quality, denoted by parameter $\phi$, has a straightforward impact on Firm $\mathbf{A}$'s open-sourcing \emph{decision}. Essentially, a larger $\phi$ implies a larger value of the open-sourcing option by accelerating the growth opportunities of open model, and a smaller value of closed-source model by increasing the growth potential of model $B$ and decreasing API profits. Therefore, the open source window size must be increasing in $\phi$. Figure \ref{fig:os_window_phi} illustrates such a relationship. The figure shows that if $\phi$ is sufficiently small, then the size of the open source window is 0, that is, $\mathbf{A}$ will not open its LLM even if it is only marginally better than $B$.

The more interesting question is how the efficiency of the open source community affects the overall model \emph{value}, after accounting for its impact on the open source decision. Figure \ref{fig:model_phi} illustrates the impact of a reduction in $\phi$, from its baseline value of 0.5 to 0.1, on the value of closed and open model for Firm $\mathbf{A}$. As anticipated, the value of open model diminishes with a decrease in $\phi$, since in the open source scenario, the efficiency of the open source ecosystem accelerates the model's development, thereby enhancing its value. Conversely, the impact of a reduction in $\phi$ on the value of closed model is more complex. When the lead of model $A$ over the alternative model $B$ is not too large, a decrease in $\phi$ also lowers the value of closed model. This decrease stems from the fact that the optimal decision for such $q$ values is to open source the model, and the value of open model influences the immediate value of closed model through the discounted option value of choosing between open and closed model, represented by $\beta \max\{ V^{O}, V^C\}$ in Equation \ref{eq:V_closed}. However, if the lead of model $A$ over $B$ is large enough that Firm $\mathbf{A}$'s optimal decision is to keep $A$ closed, then a reduction in $\phi$ essentially increases the value of closed model by limiting the competition between model $A$ and $B$, thus enabling Firm $\mathbf{A}$ to extract larger profits by selling API access to its model.

The impact of $\phi$ on $\mathbf{A}$'s model valuation presents another intriguing interpretation. The vertical dashed line in the figure marks the critical $q$ at which the value of the LLM, i.e., $\max\{V^O, V^C \}$, with $\phi=0.1$ exceeds the model value when $\phi=0.5$. Consequently, \textbf{IF} the firm could influence the efficiency of the open source ecosystem exogenously, such as by investing in the open source ecosystem or lobbying for regulations, it would opt for the former strategy when its lead over the open sourced alternative is small/moderate, and the latter when its lead is substantial.
\begin{figure}[ht]
    \centering
    \footnotesize
    \caption{Efficiency in the open source Community and LLM's Value }
    \includegraphics[width=0.75\textwidth]{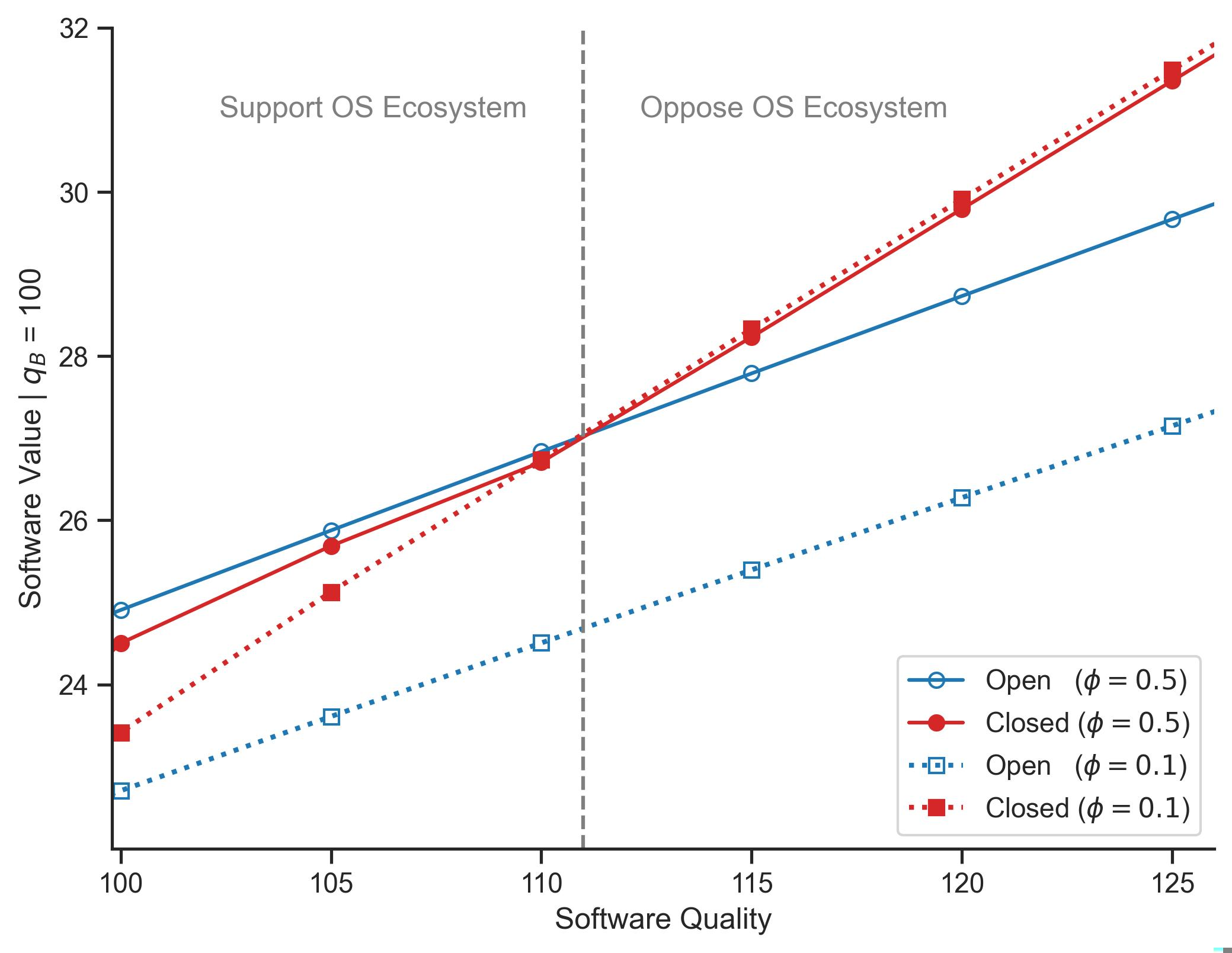} 
    
     \parbox{\textwidth}{\footnotesize{\textit{Notes:} The figure shows the valuation of closed and open model for firm $\mathbf{A}$ as a function of its model quality, under two different values of ($\phi$) denoting the efficiency of the open source community. The quality of model $B$ is held constant at 100. For further details on the modeling parameters, see Table \ref{tab:model_params}.}}
    \label{fig:model_phi}
\end{figure}
\subsubsection*{Open Source Decision and Firm Size}
Subsequently, I explore the interplay between the size of firm $\mathbf{A}$ and its decision to open source. Here, ``firm size" refers to the mass of producers in the application sector owned by the firm and possessing AI-compatible technology. Figure \ref{fig:model_size} illustrates the size of the open source window for various firm sizes. As indicated in the figure, the influence of firm size on open-sourcing decisions is not linear. For a relatively small firm, the profits from internal production are minimal, diminishing the advantages of open-sourcing, which primarily lies in the accelerated model quality growth. In such cases, revenue from API licensing becomes significantly more valuable, leading to weak incentives for open-sourcing. Conversely, for a larger firm, internal production profits become a major portion of overall profits, favoring open-sourcing due to the accelerated growth of open LLM and, consequently, the increased profits generated by Firm $\mathbf{A}$'s producers in the AS. However, when the firm becomes too large, the benefits of external contributions to open source model diminish in comparison to internal contributions, diminishing the incentives for open-sourcing once again. Overall, this results in an inverted-U shaped relationship between the firm's size and its inclination to open source its advanced LLM.

\begin{figure}[htb]
    \centering
    \footnotesize
    \caption{Firm Size and Open Soruce Decision}
    \includegraphics[width=0.75\textwidth]{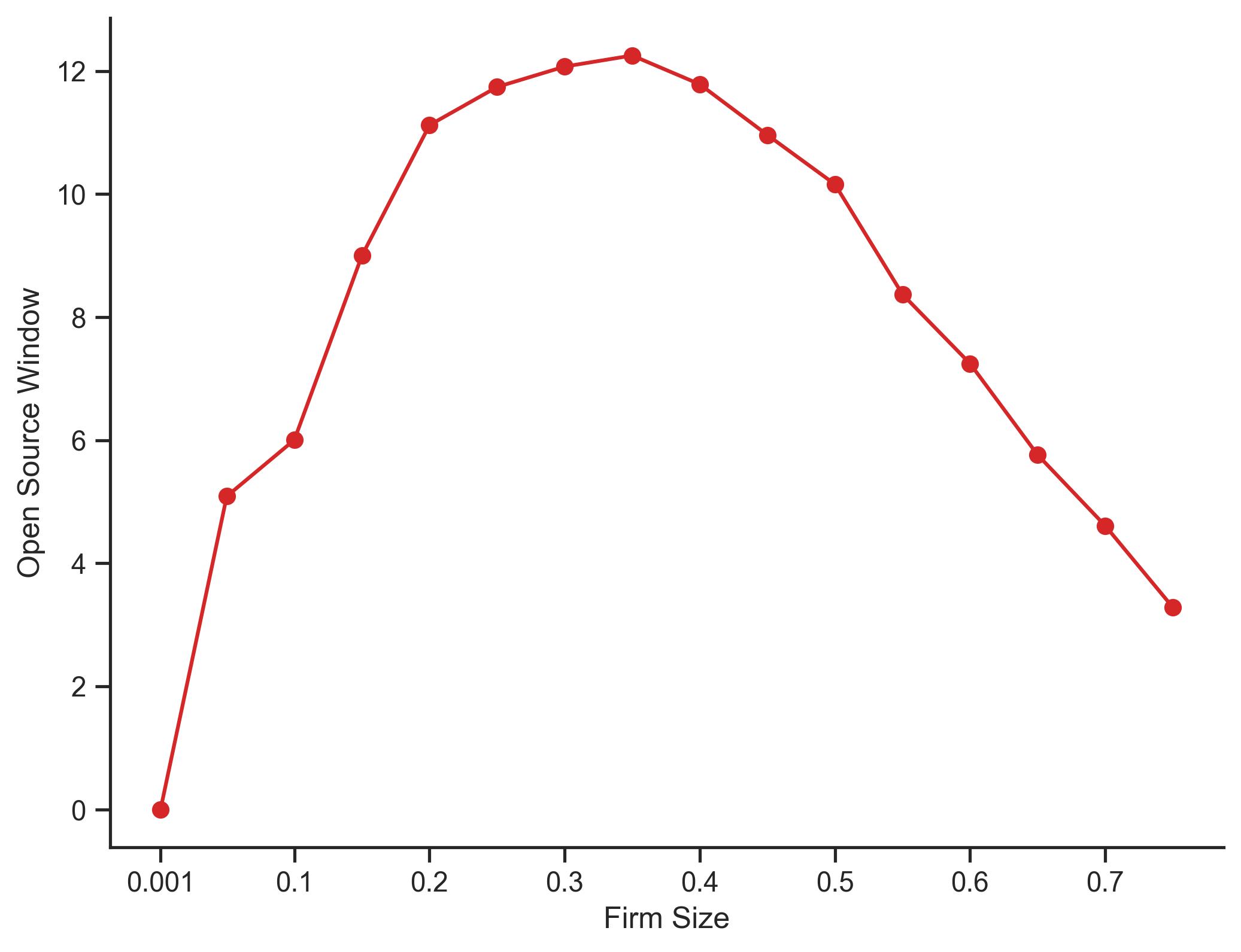} 
    
     \hspace{0.4cm}\parbox{\textwidth}{\footnotesize{\textit{Notes:} The figure illustrates the open source window size, indicating the maximum quality lead of model $A$ over model $B$ at which Firm $\mathbf{A}$ opts to open source model $A$, across different values of firm size $m$. The quality of model $B$ is fixed at 100. See Table \ref{tab:model_params} for additional details on the modeling parameters.}}
    \label{fig:model_size}
\end{figure}
\subsubsection*{Development Decision}
In the previous analysis, I assumed that the firm is initially endowed with an advanced LLM. This section delves into the decision to develop new model after observing the quality of existing open source LLM, $q_B$. Two assumptions underpin the functional form of the new model development process. First, the cost of creating a new LLM is proportional to the current open source LLM, represented as $c_D q_B$. Secondly, the expected quality of this new model is denoted by $\lambda^u q_B$, where $u$ is drawn from a uniform distribution $U[0, 1]$, and $\lambda>1$ is the development factor. Thus, the expected value of developing a new LLM is given by:
\begin{equation}
E(V = \max\{V^O, V^C\} | q_B) = E(V_A (\lambda^u q_{B}, q_B)) - c_D q_B
\label{eq:develop}
\end{equation}
where $E(V = \max\{V^O, V^C\} | q_B)$ indicates that if the firm develops a new model, its value is determined based on the observed quality and its open source rival's quality, $q_B$, leading to the selection of either open-sourcing or keeping it closed to achieve the maximum value of the open and closed options.

Figure \ref{fig:model_dev} demonstrates how the expected value of developing a new LLM varies with the quality of existing open source LLM and firm size. The figure reveals that the expected value of developing a new LLM increases with firm size and \emph{generally} diminishes with the quality of the existing open source model. This pattern suggests a notable dynamic in model development: during the early stages of technology, when existing model quality is low, both small and large firms find it beneficial to invest in developing an advanced alternative. However, as the quality of existing model improves and the investment required for new development escalates, only larger firms continue to find development profitable.
\begin{figure}[htb!]
    \centering
    \footnotesize
    \caption{LLM Development Decision}
    \includegraphics[width=.75\textwidth]{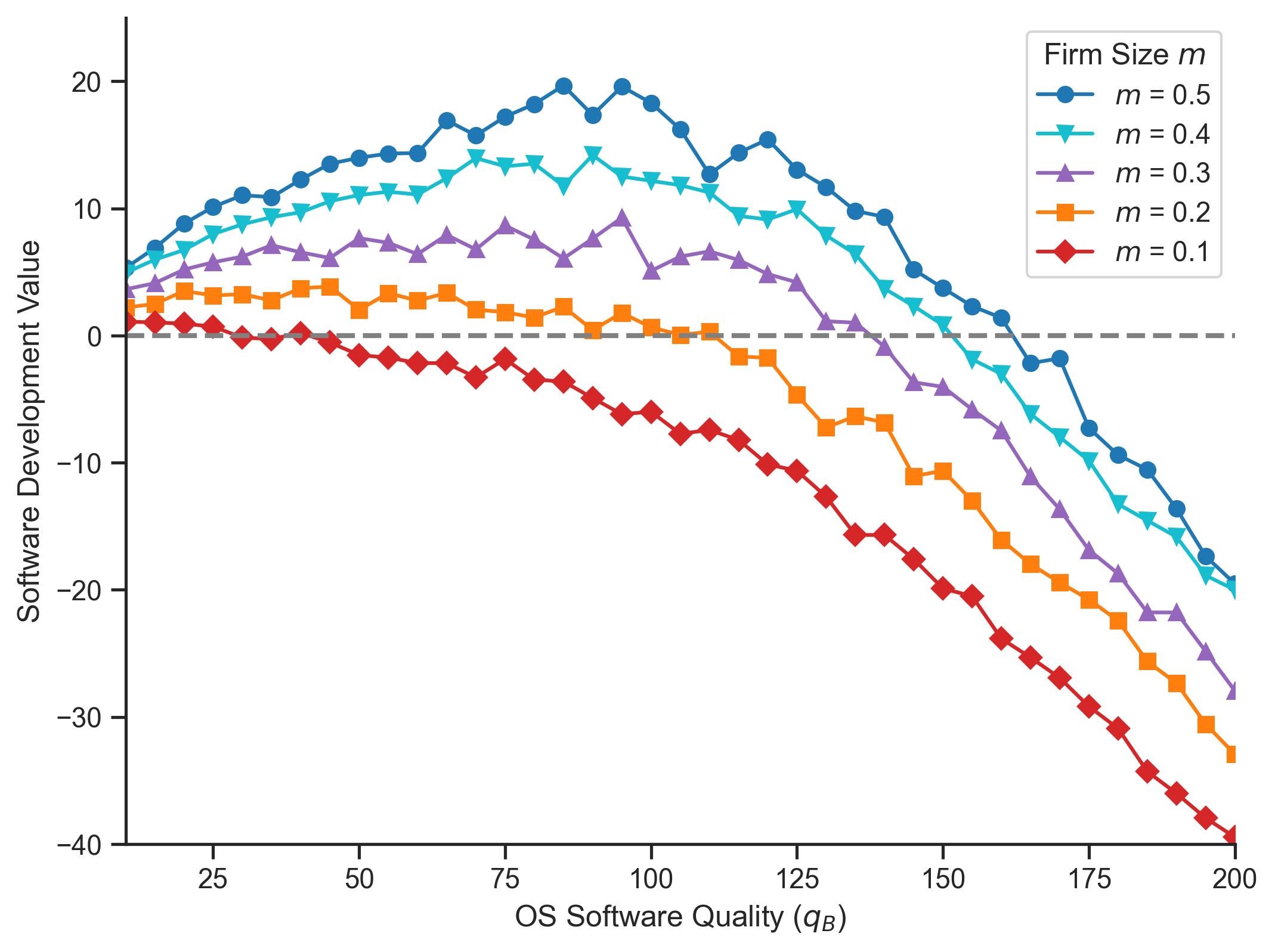} 
    
     \hspace{0.4cm}\parbox{\textwidth}{\footnotesize{\textit{Notes:} The figure depicts the expected value of developing new AI model relative to the quality of existing open source model, across various firm sizes. Additional modeling specifics are provided in Table \ref{tab:model_params}.}}
    \label{fig:model_dev}

\end{figure}

\section{Conclusion}
\label{sec:conclude}

The primary objective of this study was to explore the rationale behind for-profit companies' decisions to open source their AI software from a profit-maximizing perspective, focusing on Large Language Models (LLMs) as a particular example.

Analyzing the technology landscape using patent data reveals that LLMs are compatible with the R\&D portfolios of a wide array of firms, across varying sizes and with differentiated research technologies. Furthermore, exploiting the open source release of LLaMA, the LLM developed by Meta, the impact of open source contributions on stimulating LLM-related research activities was studied. The results suggest that contributions by LLM researchers on GitHub, considered to be a proxy for research activity, significantly increased following the release of LLaMA.

Additionally, a profit-maximizing firm's decision regarding the development and open-sourcing of a LLM as a multi-purpose technology was modeled in a dynamic discrete choice framework. The theoretical analysis yielded several compelling insights. The predictions suggested that initially, both small and large firms might find it advantageous to invest in developing new LLMs. As the development phase progresses towards the middle and late stages, only larger firms might continue to commit to new model development. Decisions about open-sourcing hinge on the model's quality advantage over open source rivals and the firm's size. A substantial gap between a firm's model and the previous open source state-of-the-art acts as a deterrent to open-sourcing, as does being excessively small or large.

Lastly, it's important to acknowledge that this study merely scratches the surface of a complex and evolving topic with significant implications for both industry and policy-making. Future research should delve deeper into areas that remain under explored in this study. Among these, the influence of competition between AI developers on their open-sourcing decisions present a critical area for exploration. Additionally, the potential impacts of regulating open source model warrant comprehensive investigation given their far-reaching consequences. Equally important is examining how open-sourcing advancements influence the behavior of downstream firms. These areas represent fertile ground for future studies, promising to enrich our understanding of the ever-changing landscape of AI development and its broader economic and societal impacts.

\clearpage
\bibliographystyle{chicago}
\bibliography{references}

\begin{thebibliography}{}

\bibitem[\protect\citeauthoryear{Agrawal, Gans, and Goldfarb}{Agrawal et~al.}{2023a}]{agrawal2023artificial}
Agrawal, A., J.~S. Gans, and A.~Goldfarb (2023a).
\newblock Artificial intelligence adoption and system-wide change.
\newblock {\em Journal of Economics \& Management Strategy\/}.

\bibitem[\protect\citeauthoryear{Agrawal, Gans, and Goldfarb}{Agrawal et~al.}{2023b}]{agrawal2023similarities}
Agrawal, A.~K., J.~S. Gans, and A.~Goldfarb (2023b).
\newblock Similarities and differences in the adoption of general purpose technologies.
\newblock Technical report, National Bureau of Economic Research.

\bibitem[\protect\citeauthoryear{Ahmed, Wahed, and Thompson}{Ahmed et~al.}{2023}]{ahmed2023growing}
Ahmed, N., M.~Wahed, and N.~C. Thompson (2023).
\newblock The growing influence of industry in ai research.
\newblock {\em Science\/}~{\em 379\/}(6635), 884--886.

\bibitem[\protect\citeauthoryear{Allen}{Allen}{1983}]{allen1983collective}
Allen, R.~C. (1983).
\newblock Collective invention.
\newblock {\em Journal of economic behavior \& organization\/}~{\em 4\/}(1), 1--24.

\bibitem[\protect\citeauthoryear{Arkhangelsky, Athey, Hirshberg, Imbens, and Wager}{Arkhangelsky et~al.}{2021}]{arkhangelsky2021synthetic}
Arkhangelsky, D., S.~Athey, D.~A. Hirshberg, G.~W. Imbens, and S.~Wager (2021).
\newblock Synthetic difference-in-differences.
\newblock {\em American Economic Review\/}~{\em 111\/}(12), 4088--4118.

\bibitem[\protect\citeauthoryear{Arora, Belenzon, Patacconi, and Suh}{Arora et~al.}{2020}]{arora2020changing}
Arora, A., S.~Belenzon, A.~Patacconi, and J.~Suh (2020).
\newblock The changing structure of american innovation: Some cautionary remarks for economic growth.
\newblock {\em Innovation Policy and the Economy\/}~{\em 20}, 39--93.

\bibitem[\protect\citeauthoryear{Arora, Belenzon, and Sheer}{Arora et~al.}{2021}]{arora2021knowledge}
Arora, A., S.~Belenzon, and L.~Sheer (2021).
\newblock Knowledge spillovers and corporate investment in scientific research.
\newblock {\em American Economic Review\/}~{\em 111\/}(3), 871--898.

\bibitem[\protect\citeauthoryear{Arts, Cassiman, and Hou}{Arts et~al.}{2021}]{arts2021technology}
Arts, S., B.~Cassiman, and J.~Hou (2021).
\newblock Technology differentiation and firm performance.
\newblock {\em Harvard Business School Strategy Unit Working Paper\/}~(22-040).

\bibitem[\protect\citeauthoryear{Beltagy, Lo, and Cohan}{Beltagy et~al.}{2019}]{beltagy2019scibert}
Beltagy, I., K.~Lo, and A.~Cohan (2019).
\newblock Scibert: A pretrained language model for scientific text.

\bibitem[\protect\citeauthoryear{Brynjolfsson, Rock, and Syverson}{Brynjolfsson et~al.}{2018}]{brynjolfsson2018artificial}
Brynjolfsson, E., D.~Rock, and C.~Syverson (2018).
\newblock Artificial intelligence and the modern productivity paradox: A clash of expectations and statistics.
\newblock In {\em The economics of artificial intelligence: An agenda}, pp.\  23--57. University of Chicago Press.

\bibitem[\protect\citeauthoryear{Business-Insider}{Business-Insider}{2023}]{nolan2023big}
Business-Insider (2023).
\newblock Big tech is inflating fears about ai's risk to humanity: Google brain cofounder.
\newblock \url{https://www.businessinsider.com/andrew-ng-google-brain-big-tech-ai-risks-2023-10}.

\bibitem[\protect\citeauthoryear{Casadesus-Masanell and Ghemawat}{Casadesus-Masanell and Ghemawat}{2006}]{casadesus2006dynamic}
Casadesus-Masanell, R. and P.~Ghemawat (2006).
\newblock Dynamic mixed duopoly: A model motivated by linux vs. windows.
\newblock {\em Management Science\/}~{\em 52\/}(7), 1072--1084.

\bibitem[\protect\citeauthoryear{Chiang, Zheng, Sheng, Angelopoulos, Li, Li, Zhang, Zhu, Jordan, Gonzalez, and Stoica}{Chiang et~al.}{2024}]{chiang2024chatbot}
Chiang, W.-L., L.~Zheng, Y.~Sheng, A.~N. Angelopoulos, T.~Li, D.~Li, H.~Zhang, B.~Zhu, M.~Jordan, J.~E. Gonzalez, and I.~Stoica (2024).
\newblock Chatbot arena: An open platform for evaluating llms by human preference.

\bibitem[\protect\citeauthoryear{CNBC}{CNBC}{2023a}]{vanian2023zuck}
CNBC (2023a).
\newblock Meta ceo mark zuckerberg touts to employees ‘incredible breakthroughs’ the company has seen in a.i.
\newblock \url{https://www.cnbc.com/2023/06/08/meta-ceo-mark-zuckerberg-talks-companys-ai-efforts-to-employees.html}.

\bibitem[\protect\citeauthoryear{CNBC}{CNBC}{2023b}]{vanian2023metas}
CNBC (2023b).
\newblock Meta's open source approach to ai puzzles wall street, techies love it.
\newblock \url{https://www.cnbc.com/2023/10/16/metas-open-source-approach-to-ai-puzzles-wall-street-techies-love-it.html}.

\bibitem[\protect\citeauthoryear{Cockburn, Henderson, and Stern}{Cockburn et~al.}{2018}]{cockburn2018impact}
Cockburn, I.~M., R.~Henderson, and S.~Stern (2018).
\newblock The impact of artificial intelligence on innovation: An exploratory analysis.
\newblock In {\em The economics of artificial intelligence: An agenda}, pp.\  115--146. University of Chicago Press.

\bibitem[\protect\citeauthoryear{Deerwester, Dumais, Furnas, Landauer, and Harshman}{Deerwester et~al.}{1990}]{deerwester1990indexing}
Deerwester, S., S.~T. Dumais, G.~W. Furnas, T.~K. Landauer, and R.~Harshman (1990).
\newblock Indexing by latent semantic analysis.
\newblock {\em Journal of the American society for information science\/}~{\em 41\/}(6), 391--407.

\bibitem[\protect\citeauthoryear{Economides and Katsamakas}{Economides and Katsamakas}{2006}]{economides2006two}
Economides, N. and E.~Katsamakas (2006).
\newblock Two-sided competition of proprietary vs. open source technology platforms and the implications for the software industry.
\newblock {\em Management science\/}~{\em 52\/}(7), 1057--1071.

\bibitem[\protect\citeauthoryear{Eloundou, Manning, Mishkin, and Rock}{Eloundou et~al.}{2023}]{eloundou2023gpts}
Eloundou, T., S.~Manning, P.~Mishkin, and D.~Rock (2023).
\newblock Gpts are gpts: An early look at the labor market impact potential of large language models.

\bibitem[\protect\citeauthoryear{Fosfuri, Giarratana, and Luzzi}{Fosfuri et~al.}{2008}]{fosfuri2008penguin}
Fosfuri, A., M.~S. Giarratana, and A.~Luzzi (2008).
\newblock The penguin has entered the building: The commercialization of open source software products.
\newblock {\em Organization science\/}~{\em 19\/}(2), 292--305.

\bibitem[\protect\citeauthoryear{Gambardella and von Hippel}{Gambardella and von Hippel}{2018}]{gambardella2018open}
Gambardella, A. and E.~A. von Hippel (2018).
\newblock Open source hardware as a profit-maximizing strategy of downstream firms.

\bibitem[\protect\citeauthoryear{Gentzkow, Kelly, and Taddy}{Gentzkow et~al.}{2019}]{gentzkow2019text}
Gentzkow, M., B.~Kelly, and M.~Taddy (2019).
\newblock Text as data.
\newblock {\em Journal of Economic Literature\/}~{\em 57\/}(3), 535--574.

\bibitem[\protect\citeauthoryear{Goldfarb, Taska, and Teodoridis}{Goldfarb et~al.}{2023}]{goldfarb2023could}
Goldfarb, A., B.~Taska, and F.~Teodoridis (2023).
\newblock Could machine learning be a general purpose technology? a comparison of emerging technologies using data from online job postings.
\newblock {\em Research Policy\/}~{\em 52\/}(1), 104653.

\bibitem[\protect\citeauthoryear{Hain, Jurowetzki, Buchmann, and Wolf}{Hain et~al.}{2022}]{hain2022text}
Hain, D.~S., R.~Jurowetzki, T.~Buchmann, and P.~Wolf (2022).
\newblock A text-embedding-based approach to measuring patent-to-patent technological similarity.
\newblock {\em Technological Forecasting and Social Change\/}~{\em 177}, 121559.

\bibitem[\protect\citeauthoryear{Henkel}{Henkel}{2004}]{henkel2004open}
Henkel, J. (2004).
\newblock Open source software from commercial firms--tools, complements, and collective invention.
\newblock {\em Zeitschrift f{\"u}r Betriebswirtschaft\/}~{\em 4}, 1--23.

\bibitem[\protect\citeauthoryear{Hovy}{Hovy}{2022}]{hovy2022text}
Hovy, D. (2022).
\newblock {\em Text analysis in python for social scientists: Prediction and classification}.
\newblock Cambridge University Press.

\bibitem[\protect\citeauthoryear{Howard and Ruder}{Howard and Ruder}{2018}]{howard2018universal}
Howard, J. and S.~Ruder (2018).
\newblock Universal language model fine-tuning for text classification.
\newblock {\em arXiv preprint arXiv:1801.06146\/}.

\bibitem[\protect\citeauthoryear{Jacobides, Brusoni, and Candelon}{Jacobides et~al.}{2021}]{jacobides2021evolutionary}
Jacobides, M.~G., S.~Brusoni, and F.~Candelon (2021).
\newblock The evolutionary dynamics of the artificial intelligence ecosystem.
\newblock {\em Strategy Science\/}~{\em 6\/}(4), 412--435.

\bibitem[\protect\citeauthoryear{Jaffe}{Jaffe}{1986}]{jaffe1986}
Jaffe, A.~B. (1986).
\newblock Technological opportunity and spillovers of r\&d: Evidence from firms' patents, profits, and market value.
\newblock {\em The American Economic Review\/}~{\em 76\/}(5), 984--1001.

\bibitem[\protect\citeauthoryear{Kaplan, McCandlish, Henighan, Brown, Chess, Child, Gray, Radford, Wu, and Amodei}{Kaplan et~al.}{2020}]{kaplan2020scaling}
Kaplan, J., S.~McCandlish, T.~Henighan, T.~B. Brown, B.~Chess, R.~Child, S.~Gray, A.~Radford, J.~Wu, and D.~Amodei (2020).
\newblock Scaling laws for neural language models.
\newblock {\em arXiv preprint arXiv:2001.08361\/}.

\bibitem[\protect\citeauthoryear{Kelly, Papanikolaou, Seru, and Taddy}{Kelly et~al.}{2021}]{kelly2021measuring}
Kelly, B., D.~Papanikolaou, A.~Seru, and M.~Taddy (2021).
\newblock Measuring technological innovation over the long run.
\newblock {\em American Economic Review: Insights\/}~{\em 3\/}(3), 303--320.

\bibitem[\protect\citeauthoryear{Lerner, Pathak, and Tirole}{Lerner et~al.}{2006}]{lerner2006dynamics}
Lerner, J., P.~A. Pathak, and J.~Tirole (2006).
\newblock The dynamics of open-source contributors.
\newblock {\em American Economic Review\/}~{\em 96\/}(2), 114--118.

\bibitem[\protect\citeauthoryear{Lerner and Tirole}{Lerner and Tirole}{2002}]{lerner2002some}
Lerner, J. and J.~Tirole (2002).
\newblock Some simple economics of open source.
\newblock {\em The journal of industrial economics\/}~{\em 50\/}(2), 197--234.

\bibitem[\protect\citeauthoryear{Lewis, Liu, Goyal, Ghazvininejad, Mohamed, Levy, Stoyanov, and Zettlemoyer}{Lewis et~al.}{2019}]{lewis2019bart}
Lewis, M., Y.~Liu, N.~Goyal, M.~Ghazvininejad, A.~Mohamed, O.~Levy, V.~Stoyanov, and L.~Zettlemoyer (2019).
\newblock Bart: Denoising sequence-to-sequence pre-training for natural language generation, translation, and comprehension.
\newblock {\em arXiv preprint arXiv:1910.13461\/}.

\bibitem[\protect\citeauthoryear{Meta}{Meta}{2023a}]{Meta2023LLaMA}
Meta (2023a, February).
\newblock Introducing llama: A foundational, 65-billion-parameter large language model.
\newblock \url{https://ai.meta.com/blog/large-language-model-llama-meta-ai/}.

\bibitem[\protect\citeauthoryear{Meta}{Meta}{2023b}]{MetaReportLLaMA}
Meta (2023b).
\newblock The llama ecosystem: Past, present, and future.
\newblock \url{https://ai.meta.com/blog/llama-2-updates-connect-2023/}.
\newblock Accessed: [insert date you accessed the site].

\bibitem[\protect\citeauthoryear{Nagle}{Nagle}{2018}]{nagle2018learning}
Nagle, F. (2018).
\newblock Learning by contributing: Gaining competitive advantage through contribution to crowdsourced public goods.
\newblock {\em Organization Science\/}~{\em 29\/}(4), 569--587.

\bibitem[\protect\citeauthoryear{Nagle}{Nagle}{2019}]{nagle2019open}
Nagle, F. (2019).
\newblock Open source software and firm productivity.
\newblock {\em Management Science\/}~{\em 65\/}(3), 1191--1215.

\bibitem[\protect\citeauthoryear{Nuvolari}{Nuvolari}{2004}]{nuvolari2004collective}
Nuvolari, A. (2004).
\newblock Collective invention during the british industrial revolution: the case of the cornish pumping engine.
\newblock {\em Cambridge Journal of Economics\/}~{\em 28\/}(3), 347--363.

\bibitem[\protect\citeauthoryear{NYT}{NYT}{2023}]{NYT2023GPT}
NYT (2023, April).
\newblock Let us show you how gpt works — using jane austen.
\newblock \url{https://www.nytimes.com/2023/04/27/upshot/gpt-from-scratch.html}.

\bibitem[\protect\citeauthoryear{Osterloh and Rota}{Osterloh and Rota}{2007}]{osterloh2007open}
Osterloh, M. and S.~Rota (2007).
\newblock Open source software development—just another case of collective invention?
\newblock {\em Research Policy\/}~{\em 36\/}(2), 157--171.

\bibitem[\protect\citeauthoryear{Post}{Post}{2023}]{de_vynck_2023}
Post, T. (2023, November).
\newblock Big tech wants ai regulation. the rest of silicon valley is skeptical.
\newblock \url{https://www.washingtonpost.com/technology/2023/11/09/ai-regulation-silicon-valley-skeptics/}.

\bibitem[\protect\citeauthoryear{Radford, Wu, Child, Luan, Amodei, Sutskever, et~al.}{Radford et~al.}{2019}]{radford2019language}
Radford, A., J.~Wu, R.~Child, D.~Luan, D.~Amodei, I.~Sutskever, et~al. (2019).
\newblock Language models are unsupervised multitask learners.
\newblock {\em OpenAI blog\/}~{\em 1\/}(8), 9.

\bibitem[\protect\citeauthoryear{Raffel, Shazeer, Roberts, Lee, Narang, Matena, Zhou, Li, and Liu}{Raffel et~al.}{2020}]{raffel2020exploring}
Raffel, C., N.~Shazeer, A.~Roberts, K.~Lee, S.~Narang, M.~Matena, Y.~Zhou, W.~Li, and P.~J. Liu (2020).
\newblock Exploring the limits of transfer learning with a unified text-to-text transformer.
\newblock {\em Journal of machine learning research\/}~{\em 21\/}(140), 1--67.

\bibitem[\protect\citeauthoryear{Rock}{Rock}{2019}]{rock2019engineering}
Rock, D. (2019).
\newblock Engineering value: The returns to technological talent and investments in artificial intelligence.
\newblock {\em Available at SSRN\/}~{\em 3427412}.

\bibitem[\protect\citeauthoryear{Spencer}{Spencer}{2003}]{spencer2003firms}
Spencer, J.~W. (2003).
\newblock Firms' knowledge-sharing strategies in the global innovation system: empirical evidence from the flat panel display industry.
\newblock {\em Strategic management journal\/}~{\em 24\/}(3), 217--233.

\bibitem[\protect\citeauthoryear{Vaswani, Shazeer, Parmar, Uszkoreit, Jones, Gomez, Kaiser, and Polosukhin}{Vaswani et~al.}{2017}]{vaswani2017attention}
Vaswani, A., N.~Shazeer, N.~Parmar, J.~Uszkoreit, L.~Jones, A.~N. Gomez, L.~Kaiser, and I.~Polosukhin (2017).
\newblock Attention is all you need.
\newblock In {\em Advances in neural information processing systems}, Volume~30.

\bibitem[\protect\citeauthoryear{von Hippel and von Krogh}{von Hippel and von Krogh}{2003}]{hippel2003open}
von Hippel, E. and G.~von Krogh (2003).
\newblock Open source software and the “private-collective” innovation model: Issues for organization science.
\newblock {\em Organization science\/}~{\em 14\/}(2), 209--223.

\bibitem[\protect\citeauthoryear{WSJ}{WSJ}{2024}]{Lin2024ShouldAI}
WSJ (2024).
\newblock Should ai be open-source? behind the tweetstorm over its dangers.
\newblock \url{https://www.wsj.com/articles/should-ai-be-open-source-behind-the-tweetstorm-over-its-dangers-65aa5c97}.

\end{thebibliography}

\newpage
\appendix 
\setcounter{figure}{0} 
\setcounter{equation}{0}
\renewcommand\thefigure{\thesection.\arabic{figure}}  
\section*{\huge Appendix}

\section{Data}
\setcounter{figure}{0} 
\label{sec:app_data}
\setcounter{table}{0} 
\renewcommand{\thetable}{\Alph{section}.\arabic{table}} 
\subsection{Data Extraction and Pre-Processing}

\subsubsection*{Papers-with-Code; arXiv}
Papers-with-Code is a community-driven initiative providing practitioners with free access to AI/ML research resources. This platform maintains up-to-date information on open-access AI/ML publications and the code repositories associated with those papers. I retrieved the data in January 2024. The raw dataset contains c.a. 219,000 publications (there are instances where a paper has entered the daset multiple times). The raw dataset does not include the publication dates and the abstracts of papers. Therefore, I use arXiv API to retrieve data of more than 142 thousand publications with unique  ``arxiv id''. For the main analysis, I use only publications linked with an official code repository on GitHub, and use the publications with unofficial repositories for training the LLM classifier in later stages.  I extract the publication date from the `published' field in the arXiv dataset, and keep only papers published from 2019 onward. The following table shows the number of publications in the dataset per year.
\begin{table}[ht]
\centering
\footnotesize
\caption{Number of publications in the arXiv dataset}
\begin{tabular}{cc}
\toprule
\textbf{Year} & \textbf{N. Papers} \\ \midrule
2019 & 11,287 \\
2020 & 17,734 \\
2021 & 23,508 \\
2022 & 26,613 \\
2023 & 28,755 \\
2024 & 290 \\ \bottomrule
\end{tabular}
\label{tab:publications}
\end{table}

\subsubsection*{GitHub}
In the first step, I retrieve the data for the remaining repositories in the Papers-with-Code dataset from GitHub using its API, of nearly 108 thousand observations, data for c.a. 106 thousand repositories were successfully retrieved. For each of these repositories, I then collect data for owners and contributors to these repositories. Data for more than 109K unique contributors and 8.2K unique organizations was successfully retrieved. 

For creating the control group, I also retrieved profiles  of data contributors to 20 popular repositories in various topics that are not directly AI-related. Each repository is the most starred non-educational repository associated with a particular ``Topic'' on GitHub. The following table presents the topic and the link of each repository in the list.
\begin{table}[ht]
\centering
\caption{Non AI-related Repositories}
\footnotesize
\begin{tabular}{ll}
\toprule
\textbf{Topic}         & \textbf{Repo}     \\ \midrule
Blockchain                 & ethereum/go-ethereum  \\
Quantum                    & Qiskit/qiskit  \\
Database                   & netdata/netdata   \\
Cloud                      & localstack/localstack \\
PHP                        & laravel/laravel    \\
JavaScript                 & vuejs/vue    \\
Android                    & flutter/flutter  \\
3D                         & mrdoob/three.js   \\
IoT                        & home-assistant/core  \\
Docker                     & moby/moby   \\
Golang                     & golang/go  \\
Microservices              & nestjs/nest \\
Django                     & django/django  \\
Rust                       & denoland/deno  \\
Game Development           & godotengine/godot   \\
Dashboard                  & grafana/grafana \\
CLI                        & ohmyzsh/ohmyzsh  \\
Vim                        & neovim/neovim \\
REST                       & tiangolo/fastapi \\
Web Applications           & angular/angular   \\ 
\bottomrule
\end{tabular}
\label{tab:control_repos}
\end{table}

\subsubsection*{Patent-Application Data}
I used patent application data files with US Patent and Trademark Office. The data were accessed through PatentsView.org in February 2024, containing related data until the end of 2023. The scope was limited to utility patent applications filed by organizations with at least two applications during 2019-2023. The dataset comprises nearly 1.4 million applications. The organizations that their name contained terms indicative of universities or research institutes were removed from the sample. The name of the organizations were cleaned and unique applicant ids were created based of the cleaned names, resulting in more than 60K applicants. The \emph{current} version of Cooperative Patent Classification System was used. 

\subsection{Feature Extraction with LLMs}
\label{subsec:parse_llm}
The scale and unstructured format of data prohibited manual or pattern oriented feature extraction. Therefore, I used LLMs to extract features for three separate tasks. First, given users profile information (bio, location, company) the LLM was asked to extract country, current organization, and sector (academia or industry) if such information is provided in the given text. Second, from the information provided for an organization (name, location, and bio) determine if the organization is a commercial entity. Finally, for creating training data, the LLM was given the title and and the abstract of the paper, and was asked to determine if the paper is related with LLMs at any capacity. For the first two tasks, GPT3.5 API was used to ensure consistency. However, for the final task, to reduce dependency on a single model, data was split between GPT3.5 and Mixtral 7x8B. Each each observation were processed  individually, and models' temperature were set to zero to make sure that the results are reproducible. The prompts are provided below.

\subsubsection*{Parsing Contributors Prompt}
\begin{lstlisting}[caption=, label=lst:prompt_contrs]
{
      "role": "system",
      "content": "You are a parser that extracts specific information from a given biographical text. You should identify whether the person works in Industry or Academia, the name of their organization, and the country they are located in, using a 2-letter country code. If any information is missing or unclear, respond with 'XX'. Output the information in JSON format with the keys 'OrgType', 'OrgName', and 'Country'."
    },
    {
      "role": "user",
      "content": f"{bio}"
    }
\end{lstlisting}

\subsubsection*{Parsing Organization Prompt}
\begin{lstlisting}[caption=, label=lst:prompt_contrs]
    {
      "role": "system",
      "content": 
        """
        Extract the name, type, and location of an organization from the provided information. 
        Remove unnecessary parts from names, such as 'the', 'inc', 'corp', 'labs', 'research', etc. 
        Classify the organization as a commercial entity (e.g., firm, startup, company, corporation) or non-commercial. 
        If information cannot be inferred, respond with 'XX'. 
        Format the output in JSON with the keys 'OrgName', 'IsCommercial' (True for commercial entities, False otherwise), and 'Country' (2-letter country codes.). 

        Example: 
        'Name: Meta Research; Location: Menlo Park, California; Bio: '
        {'OrgName': 'Meta', 'IsCommercial': True, 'Country': US}
        Process the following information:
       """
    },
    {
      "role": "user",
      "content": f"{text}"
    }
\end{lstlisting}

\subsubsection*{Parsing Papers Prompt}
\begin{lstlisting}[caption=, label=lst:prompt_contrs]
    {
      "role": "system",
      "content": 
        """
        Given the title and abstract of a research paper, determine if the paper uses, analyzes, improves, or is in any way related to Large Language Models (LLMs). 
        Return the result in a JSON format with 'LLMRelated' key is set to True if the paper is related to LLMs, and False otherwise.
        Consider the following paper:
       """
    },
    {
      "role": "user",
      "content": f"{text}"
    }
\end{lstlisting}

\clearpage
\section{Additional Results}
\setcounter{figure}{0} 
\setcounter{table}{0} 
\renewcommand{\thetable}{\Alph{section}.\arabic{table}} 
\label{sec:app_llama}

\subsubsection*{Firms Compatible with LLMs}

Figure \ref{fig:llm_sim}, Panel (a), displays the cosine similarity of the 50 companies closest to the Transformer vector in the constructed latent technology space. The firm with the highest similarity to Transformer technology is Grammarly, an English writing assistance application. The list also contains many AI startups as well as established firms such as Xiaomi, Thomson Reuters, Accenture, and PwC. An interesting observation emerges: none of the commonly recognized as ``Big Tech'' companies are present in the list\footnote{Another interesting observation is that many companies in the list have not filed for either an NLP or a Natural Language Generation application. As indicated by the corresponding CPC codes for NLP, G06F40/4, and NL Generation, G06F40/56.}. Panel (b) plots the pairwise cosine similarities of these firms' vectors in the constructed latent technology space. Although these firms are fairly close in terms of their similarity with LLMs, it appears that their overall technology portfolios are quite differentiated.
 
\begin{figure}[ht]
    \centering
    \caption{Technology Similarity with LLMs}
    \label{fig:llm_sim}
    \begin{minipage}[b]{0.45\textwidth}
        \centering
        \includegraphics[width=\textwidth]{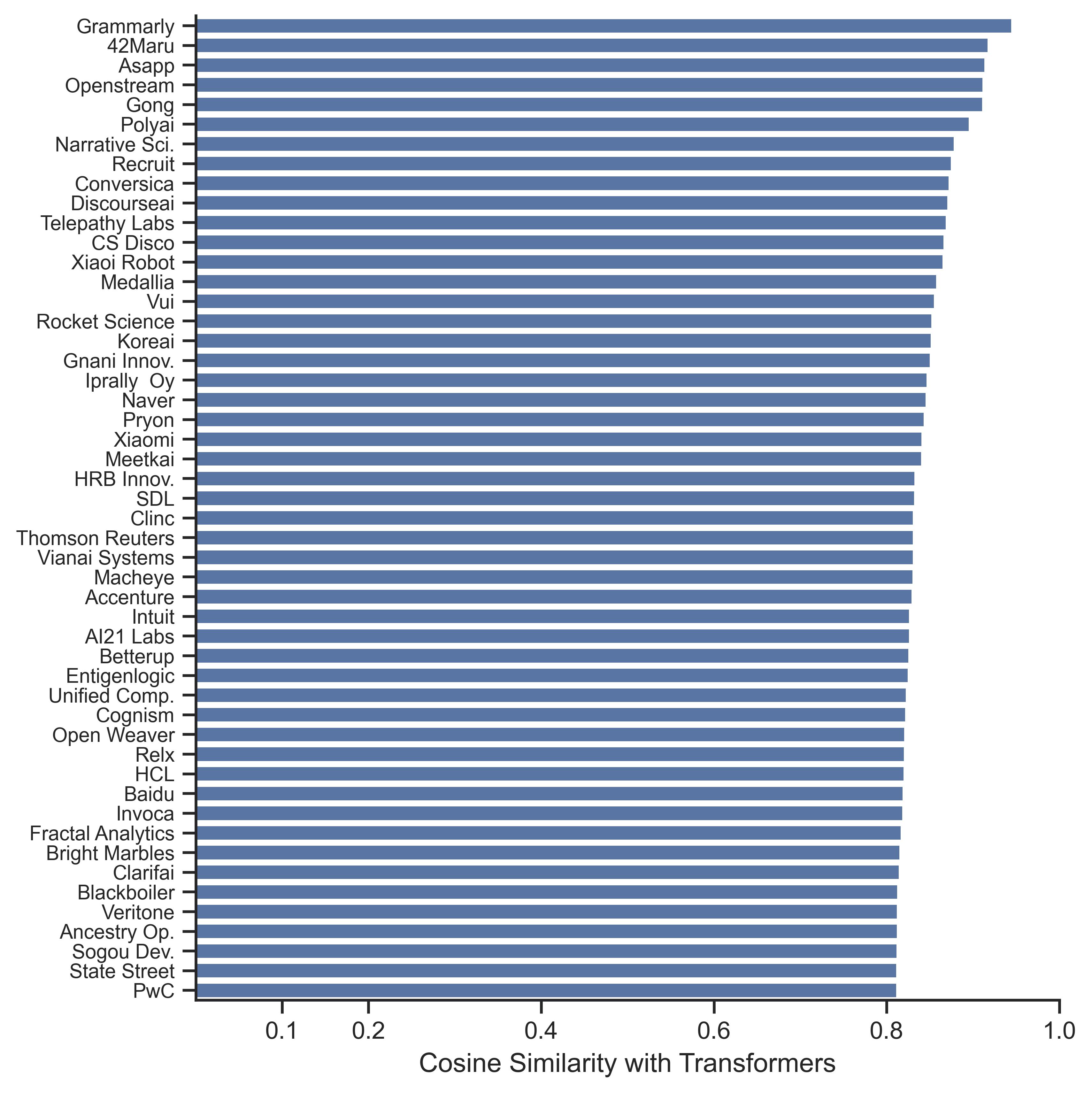}
        \caption*{(a) Similarity w. Transformers}
    \end{minipage}%
    \begin{minipage}[b]{0.45\textwidth}
        \centering
        \includegraphics[width=\textwidth]{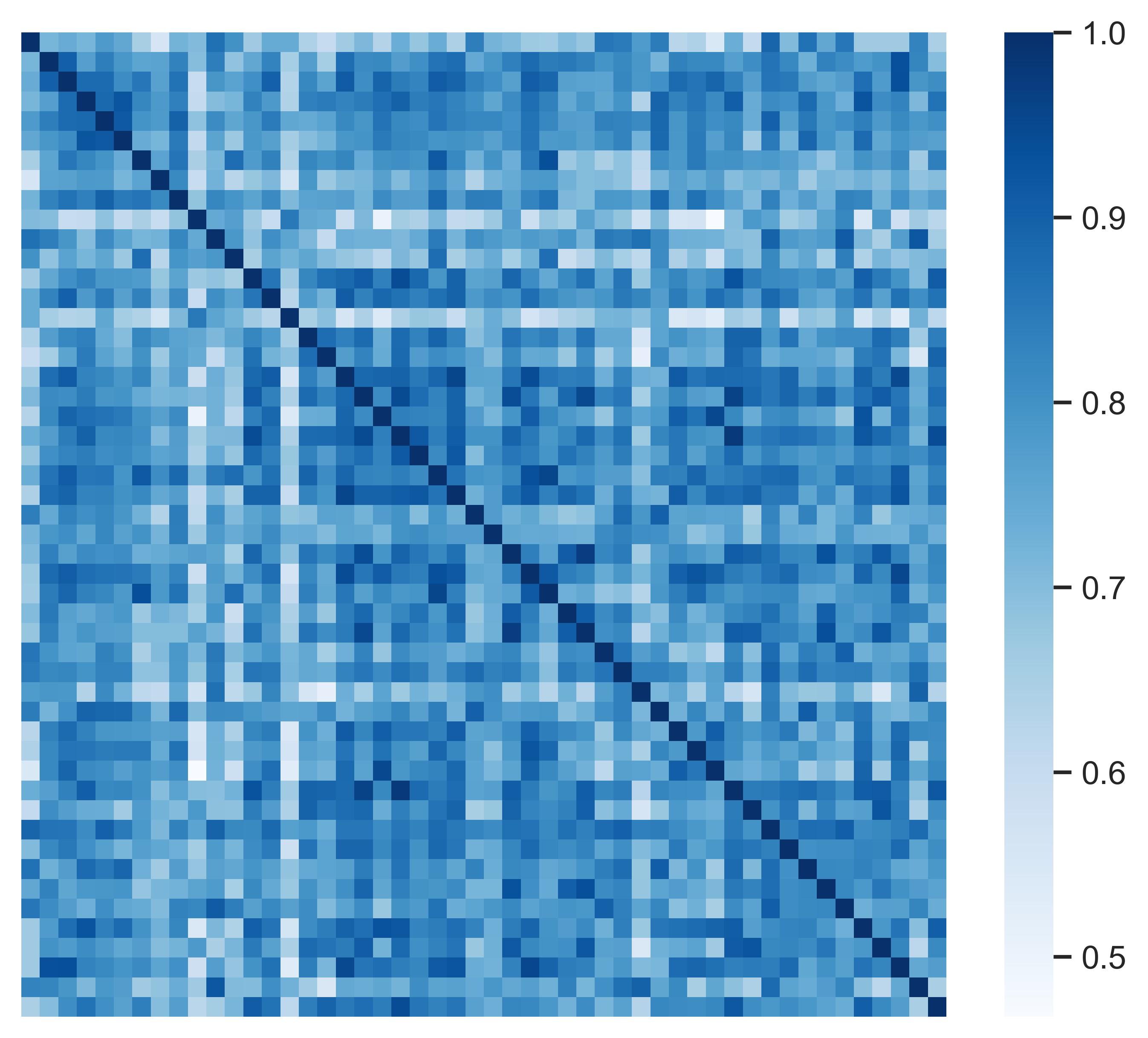}
        \caption*{(b) Cross-Technology Similarity}
    \end{minipage}
    \vspace{0.4cm}
    \parbox{\textwidth}{\footnotesize{\textit{Notes:} Panel (a) presents the 50 firms with the highest cosine similarity to Transformers in the latent technology space. Panel (b) plots the cross-cosine similarity among firms included in the left panel.}}
\end{figure}

\subsubsection*{Synthetic Difference-in-Difference Estimates}
\begin{table}[h!]
\centering
\footnotesize
\caption{SDiD Estimates of Impact of open source on GitHub Contributions}
\label{tab:reg_sdid}
\begin{threeparttable}
\begin{tabular}{lcccc}
\toprule
& All & Academy & Industry \\
\midrule
ATT & 1.109*** & 1.249*** & 0.769* \\
 & (0.159) & (0.215) & (0.417) \\
\addlinespace
N & 11,212 & 8,326 & 6,400 \\
\bottomrule
\end{tabular}
\begin{tablenotes}[para,flushleft]
\emph{Notes:} The table presents the Synthetic Difference-in-Differences estimates of the impact of LLaMA on total weekly contributions of LLM researchers on GitHub. The outcomes for Week 0 (the first seven days after LLaMA's announcement) are omitted. `Academy' indicates the group of LLM researchers whose GitHub profiles indicate that they are working in academia, and `Industry' indicates the estimates for LLM researchers whose GitHub profiles indicate they are employed in the industry. Bootstraped cluster-robust standard errors are displayed in parentheses (N=50). 
\end{tablenotes}
\end{threeparttable}
\end{table}

\subsubsection*{Simultaneous Shocks}
\begin{table}[h!]
\centering
\footnotesize
\caption{Short Term Impact Estimates with Daily Data }
\label{tab:reg_daily}
\begin{threeparttable}
    \begin{tabular}{lcccccc} 
    \toprule
     & (1) & (2) & (3) & (4) & (5) & (6) \\
     & All & All & Academy & Academy & Industry & Industry \\
     \midrule

     ATT & 0.340*** & 0.622*** & 0.309** & 0.593*** & 0.397** & 0.753*** \\
         & (0.0992) & (0.0847) & (0.120) & (0.109) & (0.165) & (0.174) \\
        \addlinespace
        Obs. & 420,250 & 420,250 & 312,543 & 312,543 & 244,032 & 244,032 \\
         $R^2$ & 0.005 & 0.002 & 0.005 & 0.002 & 0.006 & 0.002 \\
        N. Ind. & 10,250 & 10,250 & 7,623 & 7,623 & 5,952 & 5,952 \\
        
    Ind. FE & Y & Y & Y & Y & Y & Y \\
    Time FE & Y & N & Y & N & Y & N \\
    Trend   & N & Y & N & Y & N & Y \\
    Trend x Treat & N & Y & N & Y & N & Y \\ 
     \bottomrule
    \end{tabular}
\begin{tablenotes}[para,flushleft]
\emph{Notes:}  The table presents the Difference-in-Differences estimates of the impact of LLaMA on the daily  contributions of LLM researchers on GitHub, limited to 30 days before and 17 days after the introduction of LLaMA, before the release of GPT-4. The dependent variable is the relative deviation of daily contributions from their mean pre-event level. The outcomes for the first seven days after LLaMA's announcement are omitted. `Academy' indicates the group of LLM researchers whose GitHub profiles indicate that they are working in academia, and `Industry' indicates the estimates for LLM researchers whose GitHub profiles indicate they are employed in the industry. Cluster-robust standard errors are displayed in parentheses. 
\end{tablenotes}
\end{threeparttable}
\end{table}
\break
\subsubsection*{Other Treatment Variable Definitions}
\begin{table}[ht]
\centering
\footnotesize
\caption{Impact of open source on GitHub Contributions}
\label{tab:reg_treat}
\begin{threeparttable}
    \begin{tabular}{lcccccc}
    \toprule
    & (1) & (2) & (3) & (4) & (5) & (6) \\
    & All & All & Academy & Academy & Industry & Industry\\
    & LM & K-Means & LM & K-Means & LM & K-Means \\ 
    \midrule
    ATT & 1.402*** & 1.241*** & 1.389*** & 1.514*** & 1.231** & 0.880 \\
        & (0.221) & (0.246) & (0.235) & (0.265) & (0.505) & (0.558) \\
    \addlinespace
    Obs. & 199,800 & 199,800 & 154,920 & 154,920 & 120,640 & 120,640 \\
    $R^2$ & 0.006 & 0.006 & 0.007 & 0.007 & 0.004 & 0.004 \\
    N. Ind. & 9,990 & 9,990 & 7,746 & 7,746 & 6,032 & 6,032 \\
    Ind. FE & Y & Y & Y & Y & Y & Y \\
    Time FE & Y & Y & Y & Y & Y & Y \\
    Trend & N & N & N & N & N & N \\
    Trend x Treat & N & N & N & N & N & N \\
    \bottomrule
    \end{tabular}
\begin{tablenotes}[para,flushleft]
\emph{Notes:} The table presents the Difference-in-Differences estimates of the impact of LLaMA on the total weekly contributions of LLM researchers on GitHub. The dependent variable is the relative deviation of weekly contributions from their mean pre-event level. LM' denotes the group of researchers who have used Language Model' in their paper titles or abstracts. K-Means' denotes the group of researchers who have at least one paper in the cluster of NLP papers. The outcomes for Week 0 (the first seven days after LLaMA's announcement) are omitted. Academy' indicates the group of LLM researchers whose GitHub profiles indicate that they are working in academia, and `Industry' indicates the estimates for LLM researchers whose GitHub profiles indicate they are employed in the industry. Cluster-robust standard errors are displayed in parentheses.
\end{tablenotes}
\end{threeparttable}
\end{table}

\clearpage
\section{Theory Framework Appendix}
\setcounter{figure}{0} 
\label{sec:app_model}
\subsection{Model}
\subsubsection*{LLM's Demand Relation}
Recall that profit function for producer located at $x\in(m,1]$ is given by:
\begin{equation*}
    \pi_{i,t, \tau} = e^{-\gamma x_i} \left(q_{\tau,t} k_{i,t} \right)^\alpha - k_{i,t} - P_{\tau,t}
\end{equation*}
Consider the firm  located at point $x\in(m,1]$ in the AS is indifferent between paying $P$ to access model $A$'s API and using model $B$ for free. Since by assumption $q_A>q_B$ all producers located in $(m, x)$ will strongly prefer to use model $A$. Indifference condition for producer located at $x$ implies:
\begin{equation}
    \pi_A = \pi_B \Rightarrow  e^{-\gamma x} q_{A}^{\alpha} k_A^{\alpha} - k_A^\alpha - P = e^{-\gamma x} q_{B}^{\alpha} k_B^{\alpha} - k_B^\alpha
\label{eq:A1_indiff}
\end{equation}
where $k_A$ and $k_B$ are the optimal compute used when working with model $A$ and $B$, and given by:
\begin{equation*}
    k_{\tau} = \left( \alpha e^{-\gamma x} q_{\tau}^{\alpha}\right)^{1/(1-\alpha)}
\end{equation*}
where $\tau\in\{A,B\}$.

From the relations for optimal levels of compute we have $\alpha e^{-\gamma x}(q k)^{\alpha}=k$. Therefore, we can simplify Equation \ref{eq:A1_indiff} as:
\begin{equation*}
    \frac{k_A}{\alpha} - k_A -P = \frac{k_B}{\alpha} - k_B \\
    \Rightarrow P = \frac{1-\alpha}{\alpha} \left( k_A - k_B \right)
\end{equation*}
\begin{equation*}
  \Rightarrow P = \left(\frac{1-\alpha}{\alpha}\right) \left( \alpha e^{- \gamma x} \right)^{1/(1-\alpha)} \left( q_A^{\alpha/(1-\alpha)}  -  q_B^{\alpha/(1-\alpha)} \right)
\end{equation*}
\begin{equation*}
    \Rightarrow \ln \left( \frac{\alpha P}{1-\alpha} \right) = \left( \frac{1}{1-\alpha} \right) \left( \ln\alpha  - \gamma x \right)  + \ln \left( q_A^{\alpha/(1-\alpha)}  -  q_B^{\alpha/(1-\alpha)} \right)
\end{equation*}
\begin{equation*}
    \Rightarrow x = \frac{1}{\gamma} \left[ \ln \left(\alpha \left( q_A^{\alpha/(1-\alpha)}  -  q_B^{\alpha/(1-\alpha)} \right)^{(1-\alpha)}  \right) - (1-\alpha) \ln \left( \frac{\alpha P}{1-\alpha} \right) \right]
\end{equation*}
\begin{equation*}
    \Rightarrow Q_{A} = \frac{1}{\gamma} \left[ \ln \left(\alpha \left( q_A^{\alpha/(1-\alpha)}  -  q_B^{\alpha/(1-\alpha)} \right)^{(1-\alpha)}  \right) - (1-\alpha) \ln \left( \frac{\alpha P}{1-\alpha} \right) \right] - m
\end{equation*}
\subsubsection*{ Aggregate Profit Relation for Firm $\mathbf{A}$'s}
Since the transition equation in Firm $\mathbf{A}$'s dynamic problem only depends on aggregate compute used by all of its producers, the marginal profit  of any two producer it owns, w.r.t. compute $k$, must be equal, otherwise reallocation of one-unit of compute from a producer with lower marginal profit to the one with a higher marginal profit increases Firm $\mathbf{A}$'s profit. 

Therefore, consider $i$ and $j$ to be two producers owned by Firm $\mathbf{A}$ with $x_i, x_j \in [0, m]$. Following the logic provided above:
\begin{equation*}
    e^{-\delta x_i} k_i^{\alpha - 1} = e^{-\delta x_j} k_j^{\alpha - 1}
\end{equation*}

For simplicity assume $j=0$, 
\begin{equation*}
     j =0 \Rightarrow k_0 = e^{-\delta x_i/(\alpha - 1)} k_i 
     \Rightarrow k_i = e^{-\delta x_i/(1 - \alpha)} k_0
\end{equation*}

Therefore aggregate compute can be written as,
\begin{equation*}
    K = k_0 \int_{x_i=0}^{m} e^{-\gamma x_i/(1 - \alpha)} dx_i  
 = \frac{1-\alpha}{\gamma} k_0  \left[ 1 - e^{-\gamma m/(1-\alpha)} \right]
\end{equation*}
Consequently, we can write aggregate profit of Firm $\mathbf{A}$ as,
\begin{equation*}
    \Pi^F(K) = q^\alpha \int_{0}^{m} e^{-\gamma x_i} \, k_i^{\alpha} \, dx_i - K 
\end{equation*}
Or,
\begin{equation*}
\Pi^{F}(K) = q^\alpha \int_{0}^{m} e^{-\gamma x_i} \, e^{-\delta x_i/(1-\alpha)} k_0^{\alpha} \, dx_i - K
\end{equation*}
\begin{equation*}
\Rightarrow \Pi^{F}(K) = \left( \frac{\gamma K q}{(1-\alpha)(1-e^{-\gamma m /(1-\alpha)})}  \right)^{\alpha}   \int_{0}^{m} e^{-\gamma x_i} dx_i - K
\end{equation*}
After simplification, we can show that.
\begin{equation*}
    \Pi^{F}(K)= \Theta (qK)^{\alpha} - K
\end{equation*}
Where,
\begin{equation*}
    \Theta = \left( 1 - e^{-{\gamma m}/{(1-\alpha)}} \right)^{1-\alpha} \Bigg/ \left( \frac{\gamma}{1-\alpha} \right)^{1-\alpha}
\end{equation*}
\subsubsection*{Proof of Proposition 1}
\emph{Proof }:  Suppose for a given $q_B$ and other parameters of the model, there is a $q_A=q^*$ such that the value of open model $V^O(q^*)$ is equal to the value of closed model $V^C(q^*)$.  I want to show that for a small $\Delta q$, $V^C(q^* + \Delta q) > V^O(q^* + \Delta q)$ if $\Delta q>0$, and vice versa.

Suppose $\Delta q>0$, first-order approximation around $q^*$ implies, $V^C(q^* + \Delta q) \approx V^C(q^*) + \Delta q \, V^{C}_{q}(q^*)$, and  $V^O(q^* + \Delta q) \approx V^O(q^*) + \Delta q \, V^{O}_{q}(q^*)$. Since by assumption $V^C(q^*)=V^O(q^*)$,  $V^C(q^* + \Delta q) > V^O(q^* + \Delta q)$ only if $V^{C}_{q}(q^*)>V^{O}_{q}(q^*)$.

Now, let's recall the expressions for $V^O$,
\begin{equation*}
       V^{O}(q^*) = \max_{K} \left[ \Pi^F(q^*,K) + \beta \, V^{O}(q^* + \psi K + \phi K_{-A} ) \right]
\end{equation*}
Assuming flow of $K$ in any given period is small compared to stock of $q$,  
$$V^{O}(q^* + \psi K + \phi K_{-A}) \approx V^O(q^*) + (\psi K + \phi K_{-A}) V^{O}_{q}(q^*)$$
Consequently, first-order conditions (FOC) and the Envelop Theorem, imply:
\begin{equation*}
\centering
\begin{gathered}
     \text{FOC 1:} \quad  \Pi^{F}_{k}(q^*, K^O) + \beta \, \psi \, V^O_q(q^*)=0  \\
\text{Env 1.:} \quad  V^O_q(q^*)=  \Pi^{F}_{q}(q^*, K^O)
\end{gathered}
\end{equation*}

Also, for $V^C$ we have,
$$
       V^{C}(q, q_B) = \max_{K,P} \left[ \Pi^F(q,K) + \Pi^A(q,q_B,P) + \beta \max \left\{ V^{O}(q^* + \psi K) \, , V^C(q^* + \psi K \, , q_B + \phi K_{B})\right\}   \right] 
$$
After linear approximation and using $V^O(q^*)=V^C(q^*)$,
$$
       V^{C}(q, q_B) = \max_{K,P} \left[ \Pi^F + \Pi^A + \beta V^C(q^*) + \beta \max \left\{\psi K V^{O}_{q}(q^*) \, , \psi K V^{C}_{q}(q^*) + \phi K_{B} V^{C}_{q_B}(q^*) \right\}\right] 
$$
$V^C$ is decreasing with respect to $q_B$. Therefore, $V^{C}_{q_B}<0$. 

Now, assume that $V^O_q>V^C_q$.  Hence, $\max \left\{\psi K V^{O}_{q}(q^*) \, , \psi K V^{C}_{q}(q^*) + \phi K_{B} V^{C}_{q_B}(q^*) \right\} = \psi K V^{O}_{q}(q^*)$.  And we can rewrite $V^{C}(q, q_B)$ as,
$$
       V^{C}(q, q_B) = \max_{K,P} \left[ \Pi^F(q,K) + \Pi^A(q,q_B,P)+ \beta V^C(q^*) + \beta \psi K V^{O}_{q}(q^*) \right]
$$
The FOC w.r.t $K$ implies,
$$\text{FOC 2} : \Pi_K^F(q^*, K^C) + \beta \psi V^O_q(q^*)= 0
$$
However, from the FOC of open model we know: $\Pi^{F}_{k}(q^*, K^O) + \beta \, \psi \, V^O_q(q^*)=0 $.  Therefore we must have,
$$\Pi^{F}_{k}(q^*, K^O) = \Pi^{F}_{k}(q^*, K^C)  \Rightarrow K^O = K^C \Rightarrow \Pi^{F}(q^*, K^O) = \Pi^{F}(q^*, K^C)$$

From applying Envelop theorem we have,
$$\text{Env 2: }  V^C_q(q^*) = \Pi_q^F(q^*, K^C) + \Pi_q^A (q^*, q_B, P)
$$
But from Env 1 and $K^O=K^C$,
$$ V^O_q(q^*) = \Pi_q^F(q^*, K^O) = \Pi^{F}_{q}(q^*, K^C)$$
Therefore, it must be that,
$$V^C_q(q^*) = V^O_q(q^*) + \Pi_q^A (q^*, q_B, P)$$
As profit from API is increasing w.r.t $q$, we know  $\Pi_q^A (q^*, q_B, P)>0$. However, $\Pi_q^A (q^*, q_B, P)>0$ contradicts the assumption we made about $V^O_q(q^*) > V^C_q(q^*)$. Therefore, if $V^O$ and $V^C$ intersect at $q^*$, then $V^O_q(q^*) < V^C_q(q^*)$. 

Since $V^O$ and $V^C$ are both increasing functions of $q$ and $V^O_q < V^C_q$ at any point of intersection, then if $q^*$ exists, it must be unique. Moreover, for any $q>q^*$ $V^C(q)>V^O(q)$ and vice versa. $\qed$

\subsubsection*{Proof of Proposition 2}
Let's first recall the Bellman equation for closed model,
\begin{equation*}
\centering
\begin{gathered}
       V^{C}(q, q_B) = \max_{K_A,P} \left[ \Pi^F(q,K_A) + \Pi^A(q,q_B,P) + \beta \max \left\{ V^{O}(q'), V^C(q', q_B')\right\}   \right] \\
\text{s.t.} \quad q' = q + \psi K_A \\
\text{\& }   \quad q'_B = q_B + \phi K_{B} 
\end{gathered}
\end{equation*}
First, let's consider there is some $\delta>0$ such that at the optimal solution $V^C (q', q'_B) + \delta > V^O (q')$. That is the optimal solution implies that Firm $\mathbf{A}$ will keep the model closed in the subsequent period. Then, the FOC w.r.t $P$ and Envelop theorem w.r.t $q_B$ imply,
\begin{align*}
\text{FOC: } &  \Pi^A_P(q, q_B, P) + \beta \, V^C_{q_B}(q',q_B') \, \frac{\partial K_B}{\partial P} = 0 \\
\text{Env.: } & V^C_{q_B}(q,q_B,P) = \Pi^A_{{q_B}}(q,q_B, P)
\end{align*}
Substituting $V^C_{q_B}$ in the FOC from the Envelop theorem results in,
$$\Pi^A_P(q, q_B, P) + \beta \, \Pi^A_{{q_B}}(q,q_B, P) \, \frac{\partial K_B}{\partial P} = 0 $$
However, we know that profits from API is decreasing w.r.t. $q_B$. Hence, $\Pi^{A}_{q_B}<0$. Also, an increase in $P$ results in switching from model $A$ to model $B$ and therefore  an increase in $K_B$, i.e.,  $\frac{\partial K_B}{\partial P}>0$.  Therefore, for the above equality to hold at the optimal solution, we must have that,
$$\Pi^A_P(q, q_B, P) >0 $$
As $\Pi^A$ is concave w.r.t. $P$, the derivation above  implies that Firm $\mathbf{A}$ sets the price of its API below the revenue-maximizing value $P^*$ where $\Pi^A_P(q, q_B, P^*)=0$.

Conversely, let's assume that there is a $\delta'>0$ such that at the optimal solution $V^O (q') > V^C (q', q'_B) + \delta'$,  Then, the FOC w.r.t $P$  implies,
$$\Pi^A_P(q, q_B, P) = 0$$
Therefore, if the optimal choices in that space implies that Firm $\mathbf{A}$ must open it's model in the subsequent period, the firm will set the price of its API equal to it's revenue maximizing value. $\qed$
\subsection{Numerical Analysis}

In the numerical analysis of the dynamic programming model, I employed Value Function Iteration (VFI) as the primary method for solving the model. The computational work was executed using Python, with the Numba library to optimize performance. The process involved initially solving the model to determine the value of open source model. This solution then served as a foundational input to subsequently solve the model for the value of closed model. I discretized the state and control variables into intervals of equal length. Additionally, the producer's grid was discretized over the range $[0,1]$ with 1000 equally spaced points. The  parameters used in the numerical analysis of the open model are detailed in Table \ref{tab:open_params}.

\begin{table}[h] 
\footnotesize
    \centering
    \caption{VFI Parameters- Opens Model}
\begin{threeparttable}   
\begin{tabular}{lcc}
\toprule
Description &  Value  \\
\midrule
Model quality lower bound &  0  \\
Model quality upper bound &   500 \\
Model quality grid size & 101 \\
Compute lower bound &  0 \\
Compute upper bound &  20 \\
Compute grid size &  100 \\
Size grid producers &  1000 \\

\bottomrule
\end{tabular}
\end{threeparttable}
\label{tab:open_params}
\end{table}

In the analysis of the closed model, the value of the open model, derived from the preceding analysis, was used in obtaining the results. Moreover, the analysis of the closed model was substantially more demanding from computational aspects with the introduction of an additional state variables (quality model $B$) and an additional control variables (API price $P$). As a result, the computational complexity of the model substantially increases which necessitated special considerations, such as employing smaller grids for the control variables. 

Furthermore, an additional choice was introduced in the model to facilitate the  analysis for the states where quality of model $A$ was less than model $B$. This choice involved an additional option for Firm of incurring a cost proportional to $q_B$ to transition from using internal model $A$ to model $B$. However, this option only influenced the solution in scenarios where $q_A$ was substantially smaller than $q_B$, which was not the  focus of the main analysis about the open sourcing decision of the firm which was conditioned on $q_A>q_B$. 

To enhance the efficiency of the control grid usage, for any given state, I determined the maximum $P$ that rendered the producer at $x = m$ (the producer with the highest willingness to pay) indifferent between choosing model $A$ and $B$. The range from 0 to $P_m$ was then divided into 20 equal segments. Additionally, I adopted an adaptive approach to define the upper bound of the control variable $K$, based on the model parameters, ensuring that the maximum $K$ for each model configuration remained within the grid's boundaries for $K$. The modeling parameters and their specific details are outlined in Table \ref{tab:closed_params}.

For additional details on the numerical analysis of the model and insights into the nuances of its implementation, I invite readers to refer to the accompanying code.

\begin{table}[h] 
\footnotesize
    \centering
    \caption{VFI Parameters- Opens Model}
\begin{threeparttable}   
\begin{tabular}{lcc}
\toprule
Description &  Value  \\
\midrule
Model quality lower bound &  0  \\
Model quality upper bound &   500 \\
Model quality grid size & 101 \\
Compute lower bound &  0 \\
Compute upper bound &  \textit{adaptive} \\
Compute grid size & 50 \\
API price grid size & 20 \\
Size grid producers &  1000 \\

\bottomrule
\end{tabular}
\end{threeparttable}
\label{tab:closed_params}
\end{table}

\clearpage
\subsection{Additional Results}

\subsubsection*{Open Source Window Size and Quality Model $B$}
Figure \ref{fig:os_window_q_B} shows how the open source window size changes due to changes in the quality of the alternative model $q_B$. As it is displayed in the figure, the absolute size of the open source window is increasing with $q_B$. However, the relative size of the window w.r.t $q_B$ is fairly constant.
\begin{figure}[htp]
\centering
\footnotesize
\caption{Quality $B$ and open source Window $A$}
\begin{minipage}{.8\linewidth}
\label{fig:os_window_q_B}
    \centering
    \includegraphics[width=\linewidth]{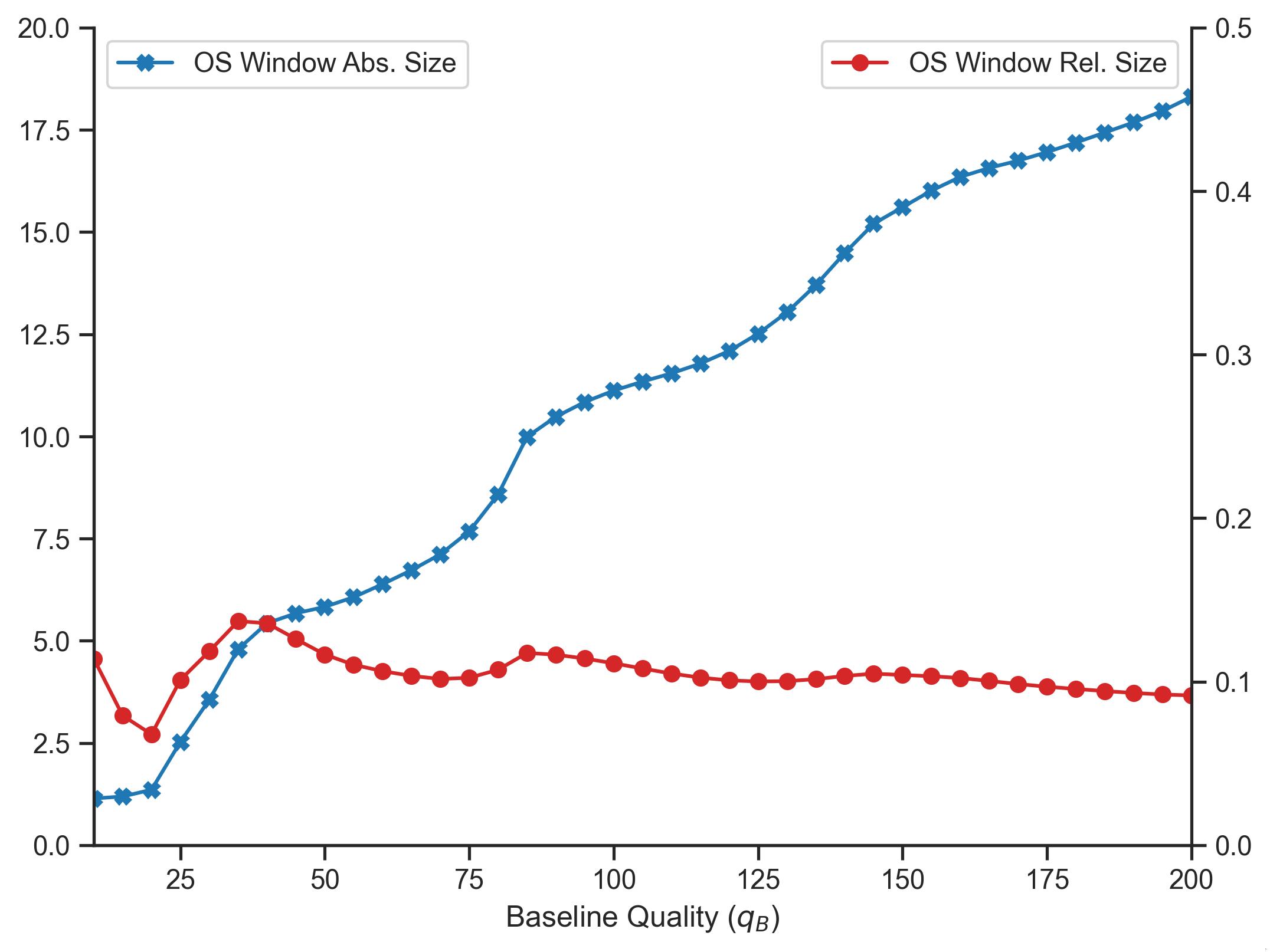}
    \parbox{\textwidth}{\footnotesize{\textit{Notes:} The figure plots the absolute and relative size of the open source window by quality of the open source model alternative. The modelling parameters used in the figure are detailed in Table \ref{tab:model_params}.}}
\end{minipage}
\end{figure}

\clearpage
\subsubsection*{Open Source Window Size and the Efficiency of open source Ecosystem}
Figure \ref{fig:osw_phi} shows how the open source window size changes due to changes in the efficiency parameter $\phi$ of the open source ecosystem. As expected, the open source window is increasing in $\phi$.
\begin{figure}[htp]
\centering
\footnotesize
\caption{Efficiency open source Community $\phi$ and open source Window}
\label{fig:os_window_phi}
\begin{minipage}{.8\linewidth}

    \centering
    \includegraphics[width=\linewidth]{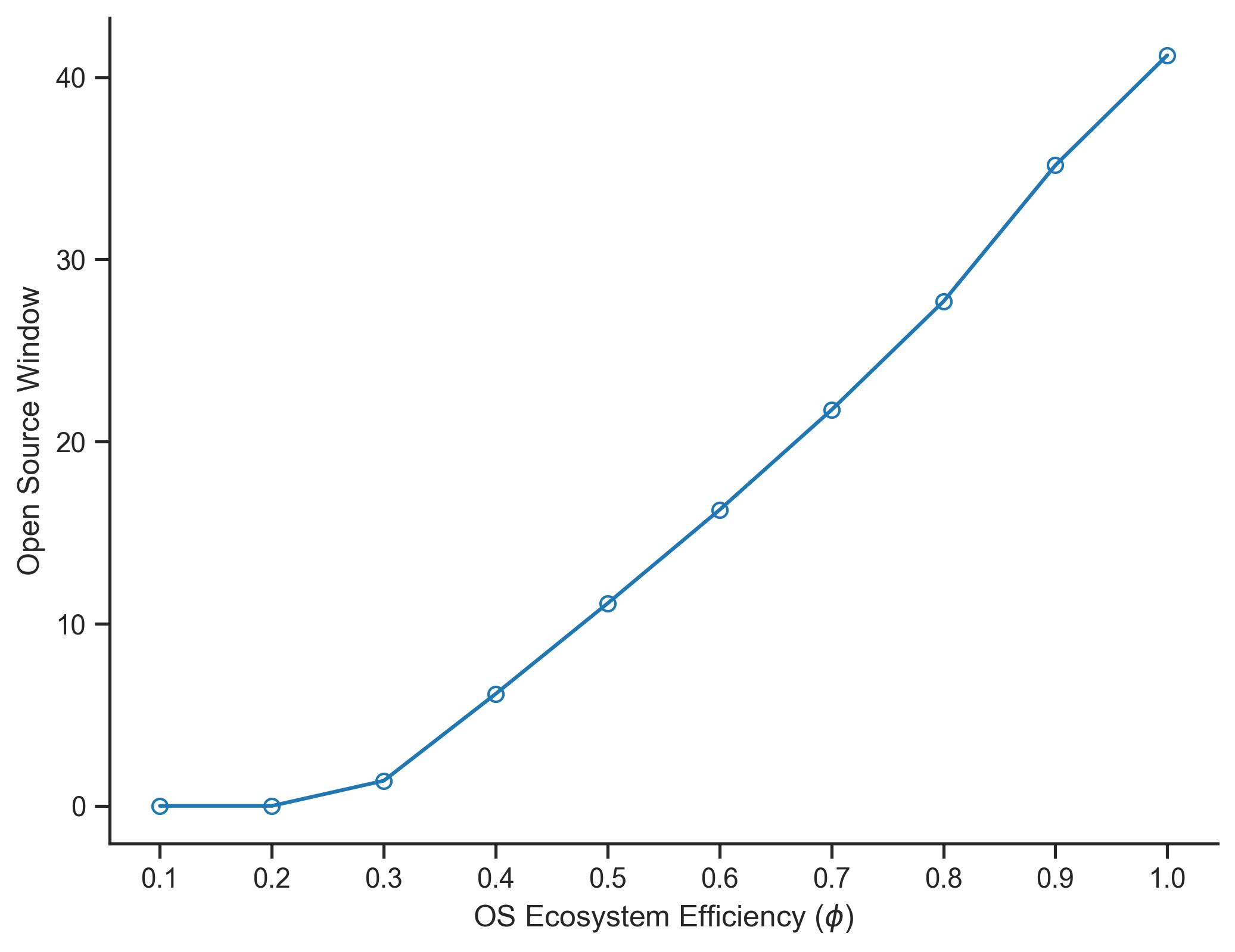}
    {\footnotesize{\textit{Notes:} The figure plots the open source window size as a function of efficiency of the open source ecosystem. The modelling parameters used in the figure are detailed in Table \ref{tab:model_params}.}}
\end{minipage}

\end{figure}

\end{document}